\documentclass[prx,twocolumn,floatfix,superscriptaddress,longbibliography,notitlepage]{revtex4-1}

\usepackage{graphicx, color, graphpap}
\usepackage{float}
\usepackage{enumitem}
\usepackage{amssymb}
\usepackage{amsthm}
\usepackage{multirow}
\usepackage[colorlinks=true,citecolor=blue,linkcolor=magenta]{hyperref}
\usepackage[T1]{fontenc}
\usepackage{bbm}
\usepackage{thmtools,thm-restate}
\usepackage{verbatim}
\usepackage{mathtools}
\usepackage{titlesec}
\usepackage{amsmath}

\usepackage[caption = false]{subfig}

\usepackage{diagbox}

\long\def\ca#1\cb{} 

\newcommand{\ketbra}[2]{| \hspace{1pt} #1 \rangle \langle #2 \hspace{1pt} |}

\newcommand{\norm}[2][]{#1| \! #1| #2 #1| \! #1|}

\newcommand{\avg}[1]{\left\langle #1\right\rangle }
\newcommand{\ket}[1]{|#1\rangle}               
\newcommand{\bra}[1]{\langle #1|}              
\newcommand{\dya}[1]{\ket{#1}\!\bra{#1}}

\newcommand{\poly}{\operatorname{poly}}

\newcommand{\CC}{\mathcal{C}}

\newcommand{\EC}{\mathcal{E}}

\newcommand{\GC}{\mathcal{G}}
\newcommand{\HC}{\mathcal{H}}

\newcommand{\LC}{\mathcal{L}}
\newcommand{\MC}{\mathcal{M}}

\newcommand{\OC}{\mathcal{O}}
\newcommand{\PC}{\mathcal{P}}

\newcommand{\SC}{\mathcal{S}}

\newcommand{\WC}{\mathcal{W}}

\newcommand{\Tr}{{\rm Tr}}

\newcommand{\Var}{{\rm Var}}

\renewcommand{\geq}{\geqslant}
\renewcommand{\leq}{\leqslant}

\newcommand{\LCb}{\overline{\LC}}
\newcommand{\LCt}{\widetilde{\LC}}

\newcommand{\Ot}{\widetilde{O}}

\newcommand{\sigmat}{\widetilde{\sigma}}

\renewcommand{\vec}[1]{\boldsymbol{#1}}  

\newcommand{\ad}{^\dagger}

\newcommand*{\id}{\openone}

\newcommand{\Sb}{\overline{S}}
\newcommand{\sbar}{\overline{s}}
\newcommand{\wbar}{\overline{w}}

\newcommand{\rhoi}{\rho_{\text{in}}}
\newcommand{\rhoo}{\rho_{\text{out}}}

\newcommand{\tout}{{\text{out}}}

\newcommand{\tin}{{\text{in}}}

\newcommand{\thv}{\vec{\theta}}

{}
{}
\newtheorem{theorem}{Theorem}
\newtheorem{lemma}{Lemma}
\newtheorem*{remark}{Remark}
\newtheorem{corollary}{Corollary}
\newtheorem{proposition}{Proposition}

\begin{document}
\title{Absence of Barren Plateaus in Quantum Convolutional Neural Networks}

\author{Arthur Pesah}
\affiliation{Theoretical Division, Los Alamos National Laboratory, Los Alamos, NM 87545, USA}
\affiliation{Department of Physics and Astronomy, University College London, London WC1E 6BT, UK}

\author{M. Cerezo}
\affiliation{Theoretical Division, Los Alamos National Laboratory, Los Alamos, NM 87545, USA}
\affiliation{Center for Nonlinear Studies, Los Alamos National Laboratory, Los Alamos, New Mexico 87544}

\author{Samson Wang}
\affiliation{Theoretical Division, Los Alamos National Laboratory, Los Alamos, NM 87545, USA}
\affiliation{Imperial College London, London, UK}

\author{Tyler Volkoff}
\affiliation{Theoretical Division, Los Alamos National Laboratory, Los Alamos, NM 87545, USA}

\author{Andrew T. Sornborger} 
\affiliation{Information Sciences, Los Alamos National Laboratory, Los Alamos, NM USA.}

\author{Patrick J. Coles}
\affiliation{Theoretical Division, Los Alamos National Laboratory, Los Alamos, NM 87545, USA}

\begin{abstract}
Quantum neural networks (QNNs) have generated excitement around the possibility of efficiently analyzing quantum data. But this excitement has been tempered by the existence of exponentially vanishing gradients, known as barren plateau landscapes, for many QNN architectures. Recently, Quantum Convolutional Neural Networks (QCNNs) have been proposed, involving a sequence of convolutional and pooling layers that reduce the number of qubits while preserving information about relevant data features. In this work we rigorously analyze the gradient scaling for the parameters in the QCNN architecture. We find that the variance of the gradient vanishes no faster than polynomially, implying that QCNNs do not exhibit barren plateaus. This provides an analytical guarantee for the trainability of randomly initialized QCNNs, which highlights QCNNs as being trainable under random initialization unlike many other QNN architectures.  To derive our results we introduce a novel graph-based method to analyze expectation values over Haar-distributed unitaries, which will likely be useful in other contexts. Finally, we perform numerical simulations to verify our analytical results. 
\end{abstract}

\maketitle

\section{Introduction}

The field of classical machine learning has been revolutionized by the advent of Neural Networks (NNs). One of the most prominent forms of classical neural network is the Convolutional Neural Network (CNN) \cite{lecun1990handwritten,lecun1998gradient}. CNNs differ from traditional fully-connected neural networks in that they use  kernels applied across network layers which reduce the dimension of the data and the number of parameters that need to be trained. The architecture of CNNs was inspired by seminal experiments investigating the structure of the visual cortex \cite{hubel1968receptive,fukushima1982neocognitron}. Indeed, CNNs have been employed for image-based tasks as their structure allows for local pattern recognition; reviews of state-of-the-art techniques can be found in Refs.~\cite{rawat2017deep,sharma2018analysis,al2017review}.

Despite the tremendous success of NNs, several difficulties were encountered during their development that hindered the trainability of the NN's parameters. While multi-layer perceptrons are more powerful than single-layer ones~\cite{haykin1994neural,minsky2017perceptrons}, their trainability is more challenging, and the backpropagation method was developed to address this issue~\cite{rumelhart1986learning}. In addition, training difficulties can also arise in NNs that employ backpropagation and gradient-based methods, as the loss function gradients can be vanishing~\cite{hochreiter2001gradient}. For CNNs this phenomena can be mitigated by choosing rectifier neuron activation functions~\cite{glorot2011deep} or by implementing layer-by-layer training.

With the advent of noisy, intermediate-scale quantum (NISQ) computers, the field of quantum machine learning has brought new approaches to solving computational problems dealing with quantum data~\cite{schuld2015introduction,biamonte2017quantum,guan2020quantum,sharma2020reformulation}. Among the most promising applications are variational quantum algorithms (VQAs)~\cite{VQE,qaoa2014,QAQC,sharma2020noise,bravo-prieto2019,cirstoiu2020variational,arrasmith2019variational,cerezo2019variational, cerezo2020variational2,bharti2021noisy}, and more generally, quantum neural networks (QNNs)~\cite{schuld2014quest,Romero,farhi2018classification,killoran2019continuous,verdon2019quantumHamiltonian,beer2020training,tacchino2020quantum,cong2019quantum}. In the near-term, these architectures employ a noisy hardware to evaluate a cost (or loss function), while leveraging the power of classical optimizers to train the parameters in a quantum circuit, or a neural network.

Recently, tremendous effort has been put forward to analyze the trainability of the cost functions of VQAs and QNNs, as it has been shown that the optimization  landscape can exhibit the so-called {\it barren plateau} phenomenon where the cost function gradients of a randomly initialized ansatz vanish exponentially with the problem size~\cite{mcclean2018barren,cerezo2020cost,sharma2020trainability,cerezo2020impact,wang2020noise,holmes2020barren,volkoff2020efficient,campos2020abrupt,marrero2020entanglement,abbas2020power}. If the cost exhibits a barren plateau, an exponentially large precision is needed to navigate through the landscape, rendering the architecture unscalable.

Specifically, it has been shown that barren plateaus can arise for different architectures such as VQAs with deep random parametrized quantum circuits~\cite{mcclean2018barren} or with global cost functions~\cite{cerezo2020cost}. Additionally, recent studies have shown that more general perceptron-based dissipative QNNs~\cite{beer2020training}  can also show barren plateau landscapes~\cite{sharma2020trainability,marrero2020entanglement}, proving that this is a general phenomenon in many quantum machine learning applications.  While strategies have been proposed to mitigate the effects of barren plateaus and improve trainability~\cite{VQSD,grant2019initialization,volkoff2020large,verdon2019learning,skolik2020layerwise,cerezo2020variational,commeau2020variational,bharti2020iterative,cervera2020meta, patti2020entanglement,patti2020entanglement}, much work still remains to be done to guarantee the efficient trainability of VQAs and QNNs.

In this work we analyze the trainability and the existence of barren plateaus in the Quantum Convolutional Neural Network (QCNN) architecture introduced by Cong et al. in~\cite{cong2019quantum}. A very similar architecture was proposed by Grant et al. in~\cite{grant2018hierarchical} under the name Hierarchical Quantum Classifier, which our work applies to as well. Motivated by the structure of CNNs, in a QCNN a series of convolutional layers are interleaved with pooling layers which effectively reduce the number of degrees of freedom, while preserving the relevant features of the input state. We remark that due to the specific architecture of the QCNN, its trainability does not follow from previous works. QCNNs have been successfully implemented for error correction, quantum phase detection~\cite{cong2019quantum, maccormack2020branching}, and image recognition~\cite{franken2020explorations}, and their shallow depth makes them a promising architecture for the near-term.

Here we provide a rigorous analysis of the scaling of the QCNN cost function gradient, under the following assumptions: (1) all the unitaries in the QCNN form independent (and uncorrelated) $2$-designs; (2) the cost function is linear with respect to the input density matrix. Under those assumptions, we show that the variance of the cost function partial derivatives is at most polynomially vanishing with the system size. This implies that the cost function landscape does not exhibit a barren plateau, and hence that the QCNN architecture is trainable under random initialization of parameters. We obtain our results by introducing a novel method to compute expectation values of operators defined in terms of Haar-distributed unitaries, which we call the Graph Recursion Integration Method (GRIM). Intuitively, our results follow from the intrinsically short depth of the QCNN (logarithmic depth), and the local nature of the cost function. We then also provide trainability guarantees for the special case where the QCNN is composed only of pooling layers. Finally, we perform numerical simulations to verify our analytical results.

Our paper is organized as follows. In Section~\ref{sec:theo} we start by reviewing the necessary theoretical background on QCNNs, barren plateaus and Haar-distributed unitaries. The GRIM is introduced in Sec.~\ref{sec:GRIM}. Section~\ref{sec:maintheorems} contains our main results in the form of Theorem~\ref{theo1} and concomitant Corollary~\ref{coro1} which show the absence of barren plateaus for randomly initialized QCNNs. Then, in Sec.~\ref{sec:pooling} we present Theorem~\ref{theo2} which analyzes the pooling-based QCNN, while Sec.~\ref{sec:numerics} contains our numerical results. Finally, our discussions and conclusion are presented in Sec.~\ref{sec:discussion}. The proof of our main results are presented in the Appendix.

\section{Theoretical Framework}\label{sec:theo}

\subsection{The QCNN architecture}

As depicted in Fig.~\ref{fig:QCNN}(a), the Quantum Convolutional Neural Network (QCNN) architecture  takes as input an $n$-qubit input state $\rhoi$, in a  Hilbert space $\HC_{\tin}$, which is sent through a circuit composed of a sequence of convolutional and pooling layers. The convolutional layer is composed of two rows of parametrized two-qubit gates acting on alternating pairs of neighboring qubits. In each pooling layer half of the qubits are measured, with the measurement outcomes controlling the unitary applied to neighbouring qubits. After $L$ convolutional and pooling layers, the QCNN  contains a  fully connected layer that applies a unitary to the remaining qubits. Finally, at the end of the QCNN one measures the expectation value of some Hermitian operator $O$.

From the previous description we have that the input state to the QCNN is mapped to a reduced state in a Hilbert space $\HC_{\tout}$ whose dimension is much smaller than that of $\HC_\tin$. The output state can  be expressed as 
\begin{equation}\label{eq:rhoout}
    \rhoo(\thv)=\Tr_{\overline{\tout}}[V(\thv)\rhoi V\ad(\thv)]\,.
\end{equation}
Here, $V(\thv)$ is the unitary that contains the gates in the convolutional and pooling layers plus the fully connected layer,  $\thv$ is the vector of the trainable parameters, and  $\Tr_{\overline{\tout}}$ denotes the partial trace over all qubits except those in $\HC_{\tout}$.  Note that the non-linearities in a QCNN arise from the pooling operators (measurement and conditioned unitary) in the pooling layers, which effectively reduce the degrees of freedom in each layer.

\begin{figure}[t]
    \centering
    \includegraphics[width=.95\columnwidth]{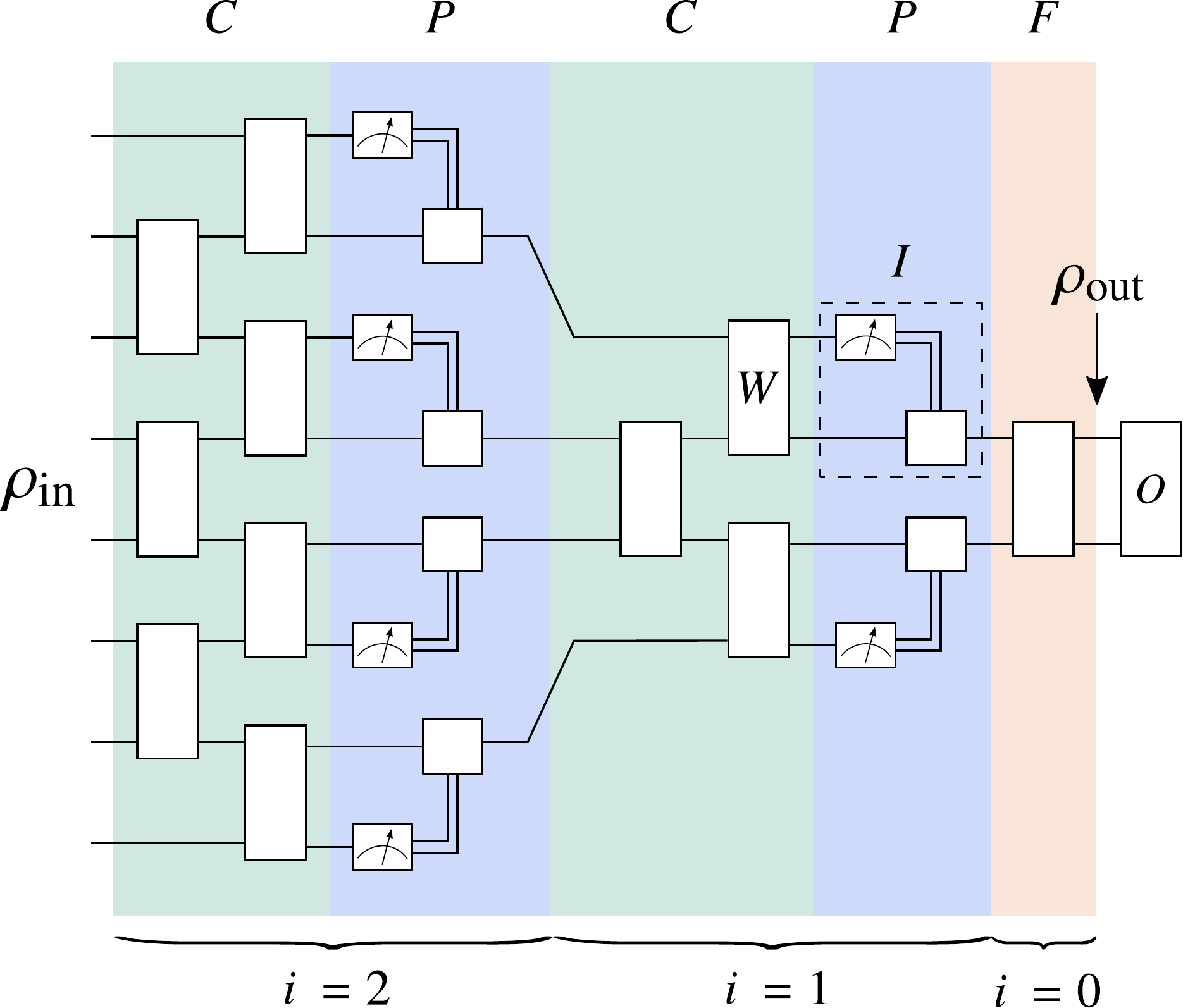}
    \caption{Schematic representation of a quantum convolutional neural network (QCNN). The input of the QCNN is an $n$-qubit quantum state $\rho_{\tin}$. The state $\rho_{\tin}$ is then sent through a sequence of $L$ convolutional (C) and pooling (P) layers. The convolutional layers are composed of two rows of two-qubit unitaries (W) acting on alternating pairs of qubits. The pooling layer is composed of pooling operators (indicated in the dashed box). In each pooling module a qubit is measured, with the measurement outcome controlling a unitary applied to a neighboring qubit (I). After the final pooling layer, one applies a fully connected unitary (F) to the remaining qubits and obtains an output state $\rho_{\tout}$ whose dimension is much smaller than that of $\rho_{\tin}$. Finally, one measures the expectation value of some operator $O$ over the state $\rho_{\tout}$.}
    \label{fig:QCNN}
\end{figure}

\subsection{Cost function} \label{sec:cost-function}

The goal of the QCNN is to employ a training set $\SC$ (of size $M=|\SC|$) containing input states $\{\rhoi^\alpha\}_{\alpha=1}^{M}$ to optimize the parameters in the QCNN and minimize a cost function which we assume can be expressed as
\begin{align}\label{eq:cost}
    C(\thv)=&\sum_{\rho^\alpha_{\tin}\in\SC} c_\alpha \Tr[\rhoo^\alpha(\thv) O ]\,,
\end{align}
where $c_\alpha$ are real coefficients, and where $\rhoo ^\alpha(\thv)$ are obtained from~\eqref{eq:rhoout} for each input state $\rhoi$. 

For example, in a binary classification problem the training set is usually of the form $\SC=\{\rhoi^\alpha,y^\alpha\}_{\alpha=1}^M$ with $y^\alpha=0$, or $y^\alpha=1$. Here we can divide $\SC$ into two subsets $\SC_0$ and $\SC_1$, of sizes $M_0=\vert \SC_0\vert$ and $M_1=\vert \SC_1\vert$ (such that $M=M_0+M_1$),  respectively composed of states associated with outputs $0$, and $1$. In this case, Eq.~\eqref{eq:cost} can be explicitly expressed as 
\begin{equation}
    C(\thv)=\sum_{\rho_{\tin}^\alpha\in \SC_0} \frac{\Tr[ \rho_{\tout}^\alpha (\thv)O]}{M_0} -  \sum_{\rho_{\tin}^\alpha\in \SC_1}\frac{\Tr[ \rho_{\tout}^\alpha(\thv)O ]}{M_1}\,,
\end{equation}
where the operator $O$ is such that for any state $\tau\in\HC_{\tout}$ we have $\Tr[\tau O]\in[0,1]$.

Here we remark that it is  convenient to rewrite the cost function as
\begin{align}\label{eq:costreduced}
        C(\thv)=&  \Tr[V(\thv)\sigma V\ad(\thv) \widetilde{O}]\,,
\end{align}
where $\widetilde{O}=\left( O  \otimes \id_{\overline{\tout}}\right)$, and $\id_{\overline{out}}$ is the identity operator over the Hilbert space $\HC_{\overline{\tout}}$ which is defined such that $\HC_{\tin}=\HC_{\tout}\otimes \HC_{\overline{\tout}}$. In addition, here we define the Hermitian operator 
\begin{equation}
    \sigma=\sum_{\rho^\alpha_{\tin}\in\SC} c_\alpha \rhoi^\alpha\,,
\end{equation}
which is a quantum state for the special case when $ c_\alpha\geq0$ for all $\alpha$, and when $\sum c_\alpha=1$. As shown below, our results are valid independently of $\sigma$ being a quantum state.

\subsection{Ansatz}

In what follows we consider for simplicity the case where $n=2^k$ and $L=\log(n)=k$, so that $\dim(\HC_{\tout})=2$.  Moreover, we assume that the unitaries in the convolutional and pooling layers are independent. That is, the convolutional and the fully connected layer in $V(\thv)$ are composed of two qubit parameterized unitary blocks acting on neighboring qubits, which we denote as $W _{ij}(\thv_{ij})$. We henceforth use the terminology of sub-layers, such that the unitaries in the first sub-layer act before those in the second-sub-layer. On the other hand, the pooling layers are composed of pooling operators $I_{ij}$. Here, $i=0,\ldots, L$ is the layer index with $i=0$ corresponding to the fully connected layer, while $j$ is the index that determines block placement within the layer.   We remark that this generalization contains as a special case the usual QCNN structure where the blocks in the same convolutional or pooling layer are identical.  Moreover, as discussed in our results section, correlating the unitaries in the convolutional layers tends to increase the magnitude of the cost function gradient.

Given a two-qubit unitary $W_{ij}(\thv_{ij})$ in a convolutional or fully-connected layer  it is common to expand it as a product of two-qubit gates of the form
\begin{equation}\label{eq:expandW}
    W_{ij}=\prod_\eta e^{-i \theta_{\eta}H_\eta}W_\eta\,,
\end{equation}
where $H_\eta$ is a Hermitian operator, $e^{-i \theta_{\eta}H_\eta}$ is a parameterized gate (such as a rotation), $W_\eta$ is an unparameterized gate (such as a CNOT), and where $\thv_{ij}=(\theta_1,\ldots,\theta_\eta,\ldots)$.  For simplicity of notation here we omit the parameter dependence in $W_{ij}$.  Note that  Eq.~\eqref{eq:expandW} can be readily used to determine the circuit description of $W_{ij}$. Moreover, if the operator $H_\eta$ has eigenvalues $\pm 1$ (or more generally, if it only has two distinct eigenvalues) one can evaluate the cost partial derivative via the parameter-shift rule~\cite{mitarai2018quantum,schuld2019evaluating}.

On the other hand, the modules in the pooling layers can be thought of as a map from a two-qubit Hilbert space $\HC_{AB}=\HC_A\otimes \HC_B$ to a single-qubit Hilbert space $\HC_B$. Given a two-qubit state $\tau\in\HC_{AB}$, the action of the measurement and the controlled unitary is given by $\tau_2=\Tr_A\left[I_{ij}\tau \right]$, where $\Tr_A$ denotes the partial trace over $\HC_A$, and where 
\begin{equation}\label{eq:isom}
 I_{ij}= \left(\Pi_0\otimes U_0^{ij} + \Pi_1\otimes U_1^{ij}\right)\,.
\end{equation}
Here $\Pi_k=\dya{i}$ for $k=1,0$ and $U_k^{ij}$  are parameterized single qubit unitaries. It is straightforward to verify that the operators $I_{ij}$ are unitaries: $I_{ij}I_{ij}\ad=I_{ij}\ad I_{ij}=\id$.

\subsection{Trainability and variance of the cost}

Consider a given trainable parameter $\theta_\mu$ in a unitary $W_{ij}$ appearing in the $\ell$-th layer of $V(\thv)$, which for simplicity we denote as $W$, and let the cost function partial derivative with respect to $\theta_\mu$ be $\partial C(\thv)/\partial \theta_\mu=\partial_\mu C$. Then, let us define the two-qubit Hilbert space where $W$ acts on as $\HC_w$, such that $\HC_{\tin}=\HC_w\otimes \HC_{\overline{w}}$.

As shown in~\cite{mcclean2018barren,cerezo2020cost,sharma2020trainability,cerezo2020impact}, the trainability of the parameters in a randomly initialized QCNN can be analyzed by studying the scaling of the variance 
\begin{equation}\label{eq:variance}
    \Var[\partial_\mu C]=\left\langle(\partial_\mu C)^2\right\rangle-\left\langle\partial_\mu C\right\rangle^2\,,
\end{equation}
where the expectation value $\langle\cdots\rangle$ is taken over the parameters in $V(\thv)$. We recall that under standard assumptions (see below) we have $\left\langle\partial_\mu C\right\rangle=0$~\cite{mcclean2018barren,cerezo2020cost}. Hence, from Chebyshev's inequality, Eq.~\eqref{eq:variance} bounds the probability that the cost function partial derivative deviates from zero by a value larger than a given $c\geq 0$ as
\begin{equation}\label{eq:Chebyshev}
    \Pr[|\partial_\mu C|\geq c] \leq \frac{\Var[\partial_\mu C]}{c^2}\,.
\end{equation}
Hence, if the variance is exponentially small (i.e., if the cost exhibits a barren plateau), the cost function gradient will be on average exponentially small and an exponential precision will be needed to navigate the flat landscape. In contrast, large variances show that no barren plateaus are present and that the trainability of the parameters under random initialization can be guaranteed.  

Let us now present the explicit form of $\partial_\mu C$. From Eqs.~\eqref{eq:costreduced} and~\eqref{eq:expandW} we find that
\begin{align}\label{eq:partialC}
    \partial_\mu C \!=\! \Tr\left[W_A V_L\sigma V_L\ad W_A\ad \left[H_\mu, W_B\ad V_R\ad \widetilde{O} V_RW_B\right]\right],
\end{align}
where we write $V=V_RWV_L$. That is, $V_R$ and $V_L$ contain all gates in the QCNN except for $W$. We henceforth omit the parameter dependency for simplicity of notation. Moreover, from the expansion of $W$ in~\eqref{eq:expandW}  we also define
 \begin{align}
     W_A=\prod_{\eta\leq\mu} e^{-i \theta_{\eta}H_\eta}W_\eta\,, \quad   W_B=\prod_{\eta>\mu} e^{-i \theta_{\eta}H_\eta}W_\eta,
 \end{align}
 so that $W=W_BW_A$. We remark that it is straightforward to verify that if $H_\mu$ is an idempotent operator, then  $\left\langle\partial_\mu C\right\rangle=0$.

\subsection{Haar-distributed unitaries}

When training the QCNN, a common strategy is to randomly initialize the parameters in $V(\thv)$~\cite{cong2019quantum}.  Hence, it is convenient to define  $\WC_{ij}$ as the set of the unitaries obtained from $W_{ij}$ for each random initialization. To analyze the trainability of the QCNN we assume that the sets $\WC_{ij}$ form independent local $2$-designs. 

Let us recall the definition of a  $t$-design. Let $\WC=\{W_y\}_{y\in Y}$ be a  set of unitaries  in a $d$-dimensional Hilbert space. Then, $\WC$ is a $t$-design if for every polynomial $P_{t,t}(W)$ of degree $t$ in the matrix elements of $W$ and of degree $t$ in those of $W\ad$, the expectation value of $P_{t,t}(W)$ over the set $\WC$ is equal to the expectation value over the Haar distribution. That is,
\begin{equation}
    \frac{1}{|Y|}\sum_{y\in Y} P_{t,t}(W_y)=\int d\mu(W) P_{t,t}(W)\,
\end{equation}
where the integral is taken with respect to the Haar measure over the unitary group of degree $d$.

The previous assumption has the following key consequence. Let us recall that the pooling layer is formed by pooling operators defined according to Eq.~\eqref{eq:isom}. Then, as  shown in  Fig.~\ref{fig:QCNN} every pooling operator $I$ is preceded by a unitary  $W$ in the second sub-layer of the convolutional layer such that their combined action can be obtained from $IW$. As $W$ forms a $2$-design it is left- and right-invariant under the action of the unitary group. That is, for any function $F(W)$, and for any unitary matrix $A$
\begin{align}\label{eq:invariance}
    \int  F(AW)d\mu(W)=&\int F(W) d\mu(W)\,.
\end{align}
Equation~\eqref{eq:invariance} shows that the action of the operators $I_{ij}$ in the pooling layer can be absorbed into the action of the unitaries in the convolutional layers as the variance of the cost function partial derivative is independent of the controlled unitaries in $I_{ij}$. Hence, as schematically shown in Fig.~\ref{fig:2}, we can consider that the unitary $V(\thv)$ of the QCNN is simply composed of two-qubit unitaries $W_{ij}$ which form independent $2$-designs. 

\begin{figure}[t]
    \centering
    \includegraphics[width=.8\columnwidth]{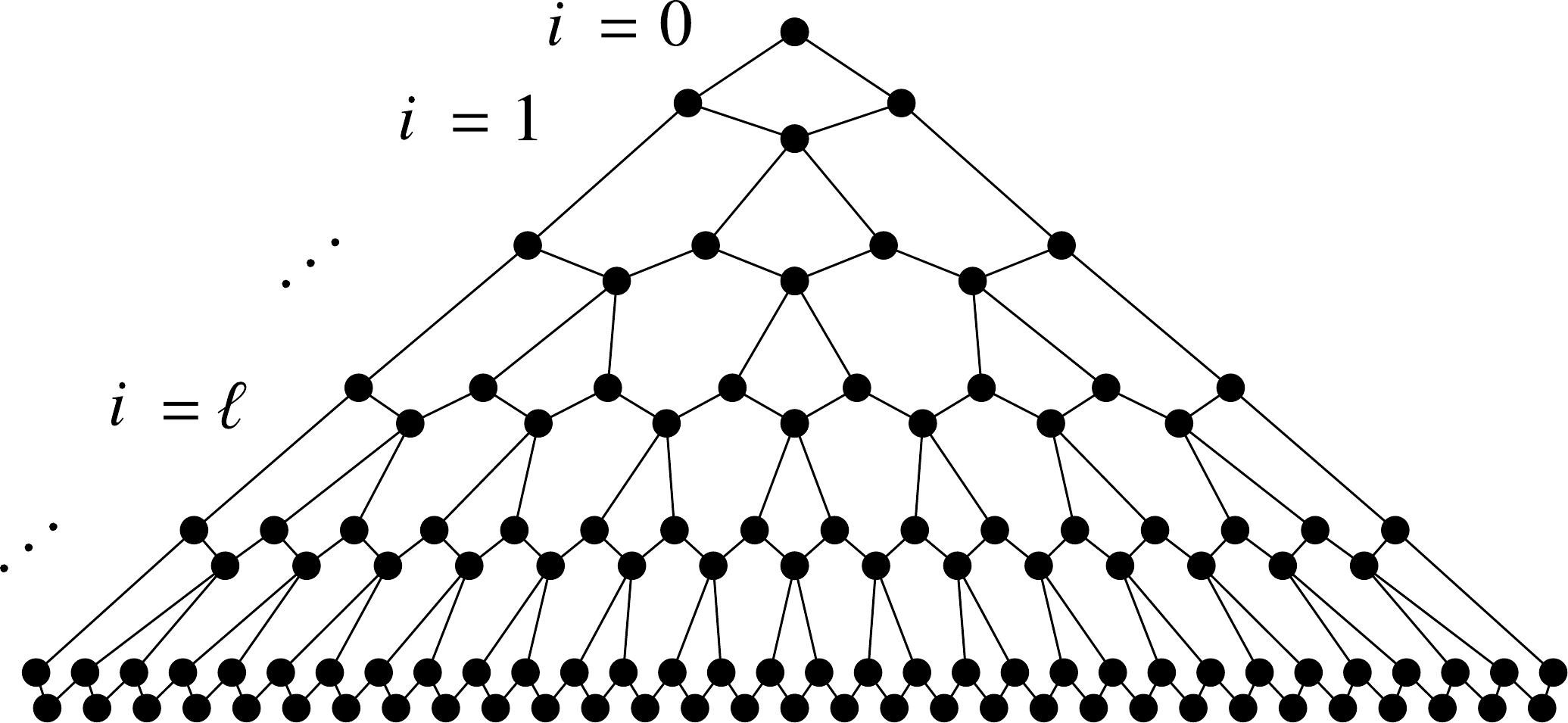}
    \caption{Tensor network representation of the QCNN. Here each unitary $W_{ij}$  in the convolutional and fully connected layers forms a $2$-design. The operators $I_{ij}$ of~\eqref{eq:isom} in the pooling layer do not modify the variance $\Var[\partial_\mu C]$ as their action can be absorbed by the action of the unitaries in the convolutional layers. Each tensor (circle) corresponds to a $W_{ij}$, with $L=0$ corresponding to the unitary in the fully connected layer. Circle with three  legs indicates that a qubit is traced out after their action.   }
    \label{fig:2}
\end{figure}

Finally, let us remark that as shown in~\cite{cerezo2020cost}, if $W_A$ and $W_B$ form independent $2$-designs, then $\langle \partial_\mu C\rangle =0$ and Eq.~\eqref{eq:variance} can be expressed as 
\begin{equation}\label{eq:Var-C}
 \Var[\partial_\mu C]=\frac{2\Tr[H_\mu^2]}{225}\sum_{\substack{\vec{p}\vec{q}\\\vec{p'}\vec{q'}}}\avg{\Delta\Omega_{\vec{q}\vec{p}}^{\vec{q'}\vec{p'}}}_{V_R}\!\!\avg{\Delta\Psi_{\vec{p}\vec{q}}^{\vec{p'}\vec{q'}}}_{V_L}\!\!.
\end{equation}
Here the summation runs over all bitstrings $\vec{p}$, $\vec{q}$, $\vec{p'}$, $\vec{q'}$ of length $2^{n-2}$. In addition, we define
\begin{align}
        \Delta\Omega_{\vec{q}\vec{p}}^{\vec{q'}\vec{p'}} &= \Tr[ \Omega_{\vec{q}\vec{p}}\Omega_{\vec{q'}\vec{p'}}] -\frac{\Tr[ \Omega_{\vec{q}\vec{p}}]\Tr[\Omega_{\vec{q'}\vec{p'}}]}{4}\,, \label{eq:Delta-Psi} \\
    \Delta\Psi_{\vec{p}\vec{q}}^{\vec{p'}\vec{q'}} &=  \Tr[\Psi_{\vec{p}\vec{q}}\Psi_{\vec{p'}\vec{q'}}]-\frac{\Tr[\Psi_{\vec{p}\vec{q}}]\Tr[\Psi_{\vec{p'}\vec{q'}}]}{4}\,,
\end{align}
 where  $\Omega_{\vec{q}\vec{p}}$ and $\Psi_{\vec{q}\vec{p}}$ are operators on $\HC_{w}$ defined as 
\begin{align}
    \Omega_{\vec{q}\vec{p}}&=\Tr_{\wbar}\left[(\ketbra{\vec{p}}{\vec{q}}\otimes\id_{w})V_R\ad \widetilde{O}V_R\right]\,,\label{eq:omegapq}
    \\
    \Psi_{\vec{p}\vec{q}}&=\Tr_{\wbar}\left[(\ketbra{\vec{q}}{\vec{p}}\otimes\id_{w})V_L \sigma V_L\ad\right]\,.
\end{align}
We use $\Tr_{\wbar}$ to denote the trace over all qubits not in $\HC_{w}$.

\section{Main results}\label{sec:main}

In this section we present our main results. First, we introduce a novel method for analyzing the scaling of the variance of Eq.~\eqref{eq:Var-C} which we call the Graph Recursion Integration Method (GRIM). Specifically, our method can be used when employing the Weingarten calculus to integrate unitaries over the unitary group. As discussed below, the GRIM is based on finding recursions when integrating groups of unitaries to form a graph which can be readily used to compute the scaling of the cost function partial derivative variance in~\eqref{eq:Var-C}. Our second main result is that we employ the GRIM to obtain a lower bound on $\Var[\partial_\mu C]$ for the QCNN architecture. Moreover, this lower bound can be used, under certain standard assumptions,  to guarantee the trainability of the QCNN.  We then present our results to guarantee the trainability of a pooling-based QCNN. Finally, we present heuristical results obtained from numerically computing $\Var[\partial_\mu C]$. 

\subsection{Graph Recursion Integration Method (GRIM)} \label{sec:GRIM}

Let us consider the task of computing the expectation values $\langle{\Delta\Omega_{\vec{q}\vec{p}}^{\vec{q'}\vec{p'}}\rangle}_{V_R}$ and $\langle{\Delta\Psi_{\vec{p}\vec{q}}^{\vec{p'}\vec{q'}}\rangle}_{V_L}$ in Eq.~\eqref{eq:Var-C}. In~\cite{cerezo2020cost} it was shown that one can solve this problem by sequentially integrating each Haar-distributed unitary  via the elementwise formula of the Weingarten calculus~\cite{collins2006integration,puchala2017symbolic} (see Appendix).

While the previous approach has been successfully used to analyze the trainability of parametrized quantum circuits and of QNNs~\cite{cerezo2020cost,sharma2020trainability} it has the difficulty that the number of terms that one has to keep track of increases exponentially with the number of unitaries integrated. The GRIM allows us to circumvent this difficulty by simultaneously integrating specific groups of unitaries and recursively grouping the resulting terms to form a {\it graph} $\GC_w$. As such, the GRIM has to be implemented to form a graph for the integration of the gates in $V_L$ and likewise for the integration of the gates in $V_R$. In what follows we consider the case of obtaining the graph to compute the expectation value of $\langle{\Delta\Omega_{\vec{q}\vec{p}}^{\vec{q'}\vec{p'}}\rangle}_{V_R}$. Moreover,  we henceforth assume that $V_R$ is composed of all the unitaries in the forward light-cone $\LC_F$ of $W$. 

The first step of the GRIM is to group the unitaries in $V_R$ into {\it modules} that cover all of  $\LC_F$. We remark that different modules will lead to different graphs, and that the fewer distinct modules one employs, the smaller the graph will be in terms of the number of distinct nodes. In Fig.~\ref{fig:all-module}(a) we show the three distinct basic modules needed to analyze the QCNN architecture, such that for any $W$ in the $\ell$-th layer, $\LC_F$ can be covered with $\ell$ such modules. Specifically,  Fig.~\ref{fig:all-module}(b) depicts an example where $\LC_F$ can be covered by three middle modules $M_{\MC}$, an edge module $M_{\EC}$, and a center module $M_{\CC}$.

\begin{figure}[t]
    \centering
    \includegraphics[width=.9\columnwidth]{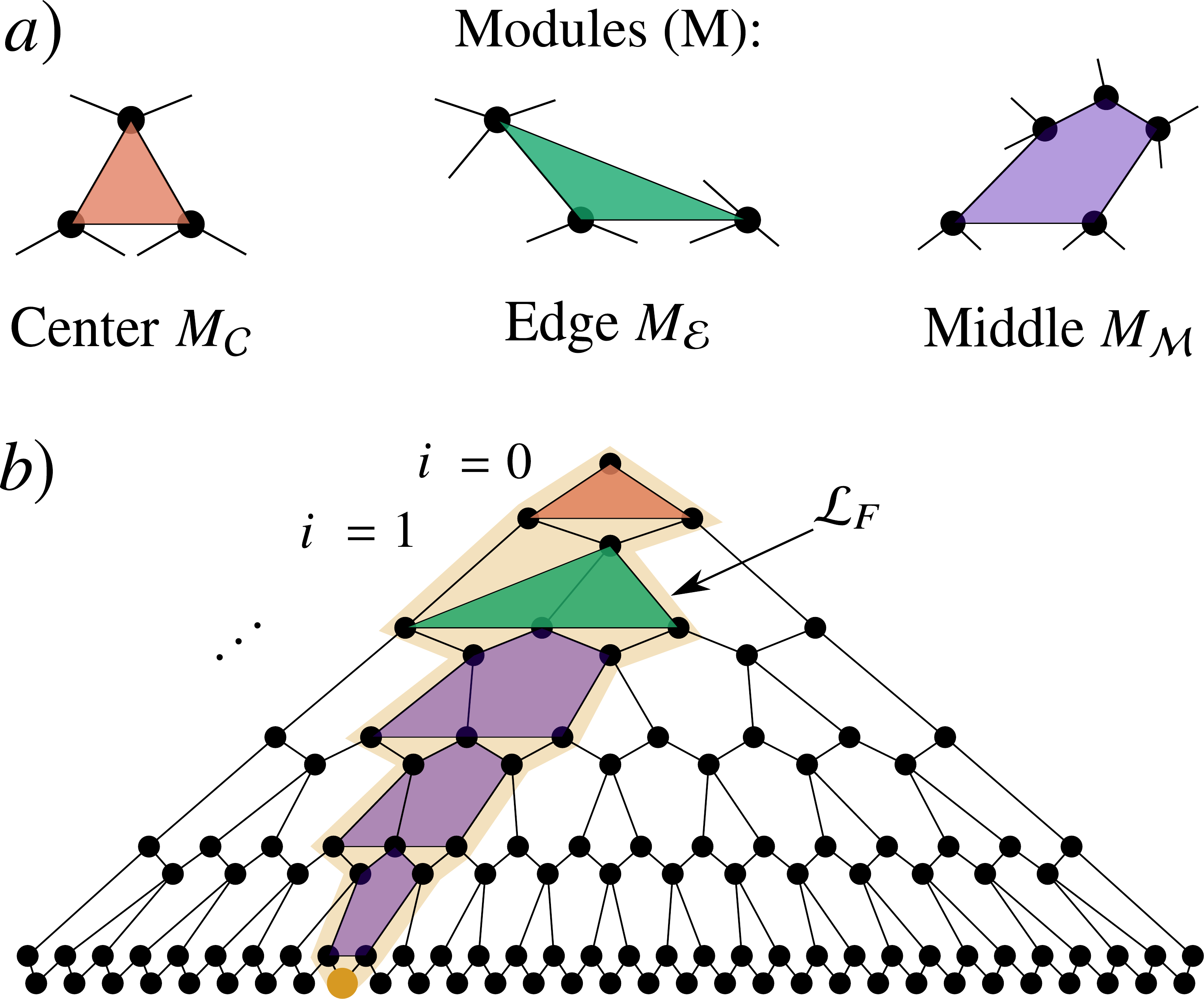}
    \caption{ GRIM modules for the QCNN architecture. (a) Basic modules $M$  from left to right, we refer to these modules as the center $M_{\CC}$, edge $M_{\EC}$, and middle $M_{\MC}$ module. Note that while both $M_{\CC}$ and $M_{\EC}$ contain three unitaries, their contractions are different. (b) Example of the first step of the GRIM. The forward light-cone $\LC_F$ of $W$ can be covered with the three basic modules. Here, $W$ is indicated with a larger colored circle.    }
    \label{fig:all-module}
\end{figure}

The next step in the GRIM is to integrate the gates in each module $M$ to form the graph. Each node $N_i\in\GC_w$ will correspond to different contractions of operators containing the modules, while the edges of the graph  will be associated with real coefficients $\lambda_{i,j}\in(0,1)$. In this notation, the coefficient $\lambda_{i,j}$ corresponds to the oriented edge $N_i\rightarrow N_j$. Moreover, the initial node can always be defined from~\eqref{eq:Delta-Psi}  as
\begin{equation}\label{eq:nodecenter}
    N_1(T)=\Tr[TT']-\frac{\Tr[T]\Tr[T']}{4}\,.
\end{equation}
where $T$ and $T'$ are two qubit operators. Following~\eqref{eq:omegapq}, here $T'$  contains primed operators obtained from $\ketbra{\vec{p}'}{\vec{q}'}$.  Then, as shown in the Appendix, the result of integrating a given module $M$ can be expressed as
\begin{equation}\label{eq:integrationmodule}
     \int_{ M}\! \prod d\mu(W_{ij})  N_i(T)=\! \sum_j e_{ij} N_j(\widetilde{T})\,,
\end{equation}
where the integral is taken over the unitaries $W_{ij}$ that appear in a module $M$, and where the terms $e_{ij}$ lead to the edge coefficients $\lambda_{i,j}$. Moreover, here the operator $\widetilde{T}$ is obtained from $T$ such that it contains the same modules as the ones in $T$ except for the one that was integrated according to~\eqref{eq:integrationmodule}. 

\begin{figure}[t]
    \centering
    \includegraphics[width=.9\columnwidth]{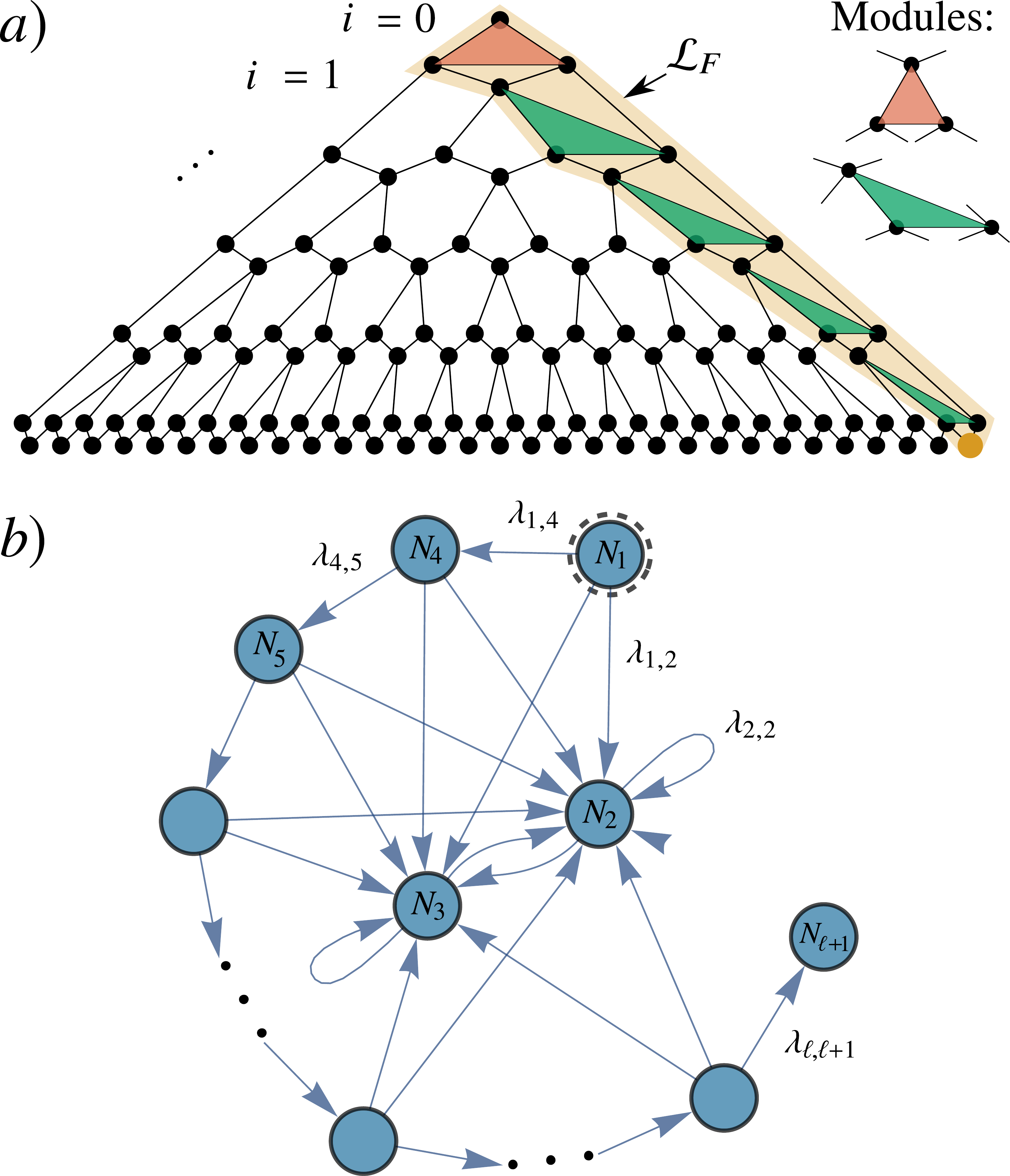}
    \caption{ GRIM  for the case when $W$ is in the edge of $\ell$-th layer of the QCNN.  (a)  The forward light-cone can be covered with $(\ell-1)$ edge modules $M_{\EC}$ and a single center module $M_{\CC}$. Here, $W$ is indicated with a larger colored circle. (b) Graph obtained from the GRIM by integrating all modules $M_{\EC}$ in ar$\LC_F$ for the case in (a). Graph $\GC_w$ obtained via GRIM. Here the edge coefficients $\lambda_{i,j}$ are real and positive numbers such that    $\lambda_{i,j}\in(0,1)$. Moreover, for this case the number of nodes is $(\ell+1)$ and the structure of the graph is such that that one can easily obtain the graph for any value of $\ell$.  }
    \label{fig:edge-module}
\end{figure}

Equation~\eqref{eq:integrationmodule} is at the basis of the GRIM as it shows that when a module is integrated one connects the nodes $N_i\rightarrow N_j$ for all $j$ in the summation. By sequentially integrating all unitaries, one can form an oriented graph which always starts with  $N_1$ according to~\eqref{eq:Delta-Psi}. As discussed in the Appendix, the number of nodes in the graph, denoted as $|\GC_w|$, is always in $\OC(\ell)$ for the QCNN architecture. Additionally, the structure of the graphs $\GC_w$ is recursive, which allows us to obtain results valid for arbitrary $\ell$ and for all placements of the unitary $W$. In Fig.~\ref{fig:edge-module}(a) we show an example where $W$ is in the edge of the $\ell$-th layer, and where the forward light-cone $\LC_F$ is composed of $(\ell-1)$ edge modules and a single center module. The graph  that arises from integrating the unitaries in the $M_{\EC}$ modules is shown in Fig.~\ref{fig:edge-module}(b). We refer the reader to the Appendix for additional details on the derivation of this graph and the definition of the $N_i$ nodes. For this specific case the graph contains $|\GC_w|=(\ell+1)$ nodes.

Once the graph  $\GC_w$ has been formed, one can use it to explicitly compute the contribution of expectation values  in the variance of~\eqref{eq:Var-C}. Specifically, the following Proposition, proved in the Appendix, holds
\begin{proposition}\label{prop1}
Let $W$ be a unitary in the $\ell$-th layer of the QCNN. The contribution of $\langle{\Delta\Omega_{\vec{q}\vec{p}}^{\vec{q'}\vec{p'}}\rangle}_{V_R}$ to the variance can be computed via the GRIM as
\begin{align}\label{eq:finalresult}
     \Var[\partial_\mu C]\geq\frac{\Tr[H_\mu^2] \varepsilon_O}{450}
    \avg{\varepsilon_{\widetilde{\sigma}_w}}_{V_L}
    \sum_{\vec{g}\in P_\ell(\GC_w)} \Lambda_{\vec{g}} \,,
\end{align}
where $\widetilde{\sigma}_w=\Tr_{\overline{w}}[V_L \sigma V_L\ad]$ is the reduced operator $\sigma$ in the subsystem of the qubits that $W$ acts on, and where we defined for an operator $O$ 
\begin{equation}
    \varepsilon_O=D_{HS}\left(O,\Tr[O]\frac{\id}{4}\right)\,,
\end{equation}
with $D_{HS}(A,B)=\Tr[(A-B)^2]$ the Hilbert-Schmidt distance between $A$ and $B$. Here $\GC_w$ denotes the oriented graph obtained by integrating all the modules in $\LC_F$, $P_\ell(\GC_w)$ is the set of all paths of length $\ell$ over $\GC_w$, and $\vec{g}=(g_1,\ldots,g_\ell)$ is a vector which indicates the nodes in the path such that $g_i\in[1,\ldots,|\GC_w|]$. Finally, $\Lambda_{\vec{g}}$ are real positive numbers obtained as
\begin{equation}
    \Lambda_{\vec{g}}=\prod_{\xi=1}^\ell\lambda_{g_\xi,g_{\xi+1}}\,,
\end{equation}
where $\lambda_{g_\xi,g_{\xi+1}}\in(0,1)$ are the edge coefficients associated with the oriented  edge $N_\xi\rightarrow N_{\xi+1}$. 
\end{proposition}

Note that $\varepsilon_O=N_1(O)$ from Eq.~\eqref{eq:nodecenter} by taking $T=T'$, meaning that  here we are using the surprising fact that the final node is {\it always}  $N_1(O)$. Proposition~\ref{prop1} shows that given a unitary $W$, the contribution from the expectation value of $\Delta\Omega_{\vec{p}\vec{q}}^{\vec{p'}\vec{q'}}$ can be ultimately obtained by adding up all the paths in the oriented graph $\GC_w$ weighed by the coefficients associated with the edges taken over the path.   Moreover, since $\varepsilon_O$ and $\varepsilon_{\widetilde{\sigma}_w}$ are always positive,   any single path over the graph leads to a lower bound for $\Var[\partial_\mu C]$. Finally, we remark that one can also employ the GRIM to  integrate the gates in $V_L$ to compute the expectation value $\avg{\varepsilon_{\widetilde{\sigma}_w}}_{V_L}$. 

\begin{figure}[t]
    \centering
    \includegraphics[width=.9\columnwidth]{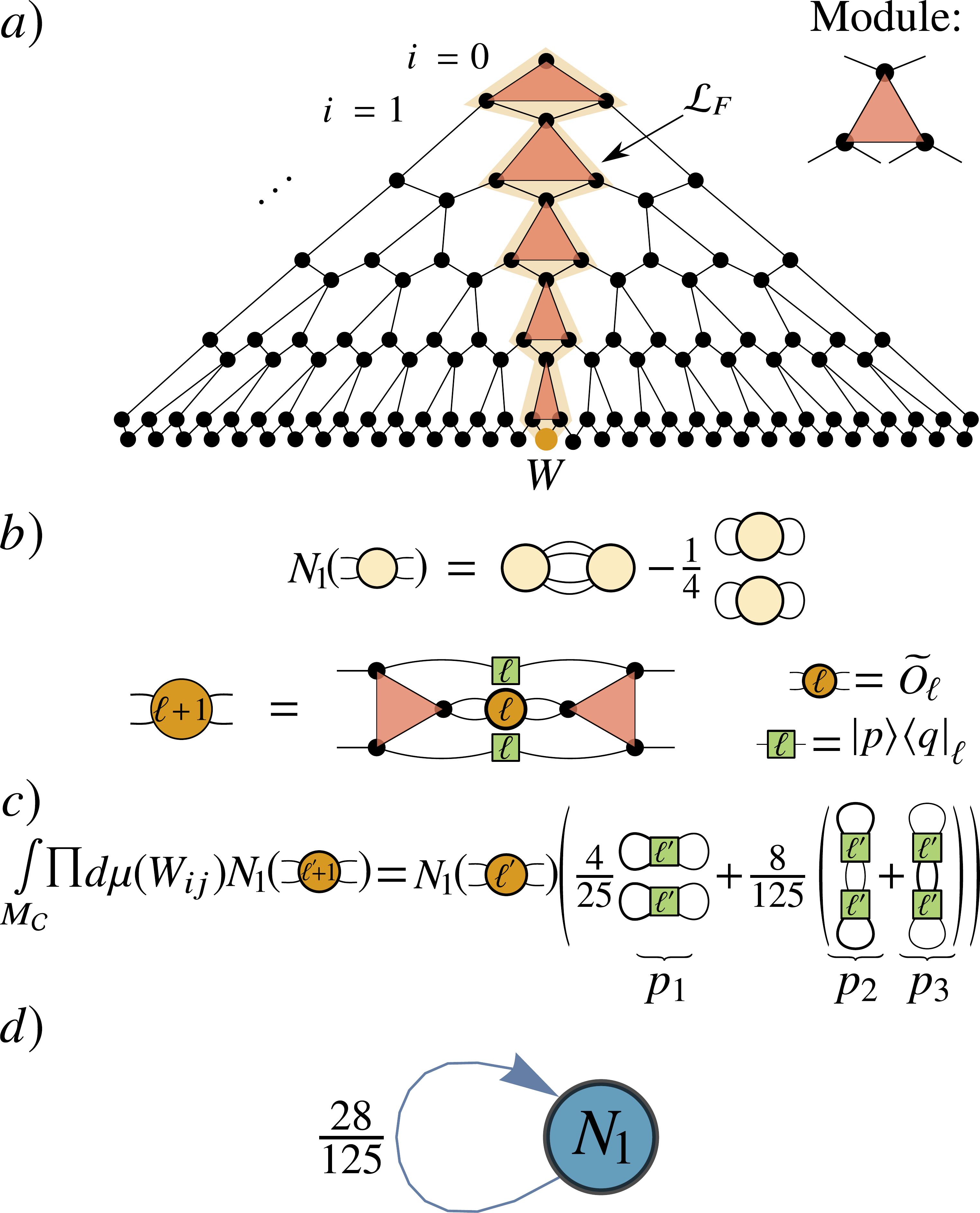}
    \caption{GRIM  for the case when $W$ is in the center of $\ell$-th layer of the QCNN. (a)  The unitaries in $\LC_F$ can be covered by $\ell$ center modules $M_{\CC}$.  (b) Schematic representation of the tensor contraction $N_1$ of Eq.~\eqref{eq:nodecenter}. Here we also show the recursive relation between the operators $\widetilde{O}_\ell$.  (c) Tensor network representation of~\eqref{eq:integration} obtained by integrating the $(\ell'+1)$-th center module. The operators $p_2$ and $p_3$ correspond to contractions over different qubits, as indicated by the bold lines. (d) Graph $\GC_w$ obtained from the GRIM for the case in (a). }
    \label{fig:center-module}
\end{figure}

In what follows we briefly give an example of how to use the GRIM to compute the expectation value $\langle{\Delta\Omega_{\vec{q}\vec{p}}^{\vec{q'}\vec{p'}}\rangle}_{V_R}$ for the case when $W$ is the center unitary of the $\ell$-th layer of the QCNN as indicated in Fig.~\ref{fig:center-module}(a). Here we can see that $\LC_F$ can be covered with $\ell$ center modules $M_{\CC}$. Employing the notation $\widetilde{O}_\ell=\Omega_{\vec{q}\vec{p}}$ it is straightforward to see that $\langle{\Delta\Omega_{\vec{q}\vec{p}}^{\vec{q'}\vec{p'}}\rangle}_{V_R}=\langle N_1(\Ot_\ell)\rangle_{V_R}$. Then, as shown in Fig.~\ref{fig:center-module}(b) in the tensor network representation of quantum circuits,  we can always write the recursion relation
\begin{equation}
    \widetilde{O}_{\ell}=\Tr_{\overline{w}}[(\ketbra{\vec{p}}{\vec{q}}_{\ell-1}\otimes\id_{w})M_{\CC}\widetilde{O}_{\ell'-1}M_{\CC}\ad]\,.
\end{equation}
Here, $\ketbra{\vec{p}}{\vec{q}}_{\ell-1}$ denotes the projectors on the qubits that the $\ell$-th module $M_{\CC}$ acts on. With this notation it then follows that $\widetilde{O}_{0}=O$.

By employing the elementwise formula of the Weingarten calculus to integrate unitaries over the Haar distribution (see Appendix) we find for all $\ell' \leq \ell$
\begin{equation}\label{eq:integration}
    \int_{\!M_{\CC}}\!\!   N_1(\widetilde{O}_{\ell'}) d\mu= \frac{4N_1(\widetilde{O}_{\ell'-1})}{25}\left(p_1+\frac{2(p_2+p_3)}{5}\right)\,,
\end{equation}
where the integration is over all unitaries $W_{ij}$ in the $\ell'$-th module $ M_{\CC}$. As depicted in  Fig.~\ref{fig:center-module}(c) the coefficients $p_i$ with $i=1,2,3$  arise from different contractions of $\ketbra{\vec{p}}{\vec{q}}_{\ell'-1}$. Then, as shown in the Appendix, the operators $p_2$ and $p_3$ lead to a factor $1/2$, while the operator $p_1$ leads to a factor of $1$ so that $\frac{4}{25}(e_1+\frac{2(e_2+e_3)}{5})=28/125$. Equation~\eqref{eq:integration} shows how the structure of the graph $\GC_w$ emerges, as integrating a module connects the node $N_1$ with itself with a coefficient $\lambda_{1,1}=28/125$. The ensuing graph is presented in Fig.~\ref{fig:center-module}(d), and we can see that it is composed of a single node. 

\subsection{Trainability of the QCNN}\label{sec:maintheorems}

Here we present the main result of our article, which is derived using the GRIM. The detailed proofs can be found in the Appendix. For convenience of notation we here consider the case when $W$ is in the first sub-layer. The case where it belongs to the second sub-layer is considered in the Appendix and we remark that the results are substantially the same for both cases, albeit with some notational differences.  

\begin{theorem}\label{theo1}
Consider the QCNN cost function of Eq.~\eqref{eq:costreduced} and let $W$ be a unitary in the first sub-layer of the $\ell$-th layer of $V(\thv)$. Moreover, let us assume  that the unitaries $W_{ij}$, $W_A$ and $W_B$ in the convolutional and fully connected layers form independent $2$-designs. Then, the lower bound of the variance of the cost partial derivative with respect to a parameter $\theta_\mu$ in  unitary $W$ can be computed, using the GRIM, from the graph $\GC_w$ as
\begin{align}
    \Var[\partial_\mu C]\geq F_n(L,\ell)\,,
\end{align}
with
\begin{equation}\label{eq:functionlower}
    F_n(L,\ell)=\frac{1}{9}\frac{\Tr[H_\mu^2] \varepsilon_O \varepsilon_{\sigma_{w}} }{50^{L-\ell+1}} \sum_{\vec{g}\in P_\ell(\GC_w)} \Lambda_{\vec{g}}.
\end{equation}
Here $\sigma_w=\Tr_{\overline{w}}[\sigma]$ is the reduced operator $\sigma$ in the subsystem of the qubits that $W$ acts on. Moreover, $\varepsilon_O$, $P_\ell(\GC_w)$, $\vec{g}$, and $\Lambda_{\vec{g}}$ are defined in Proposition~\ref{prop1}.
\end{theorem}


From Theorem~\ref{theo1} we can obtain the following corollary.

\begin{corollary}\label{coro1}
Consider the function $F_n(L,\ell)$ defined in Eq.~\eqref{eq:functionlower}. Then, since $L$ is at most $\OC(\log(n))$, the variance $\Var[\partial_\mu C]$ is at most polynomially vanishing with $n$  as  
\begin{equation}\label{eq:corlower}
    F_n(L,\ell)\in\Omega\left(\frac{1}{\poly(n)}\right)\,, 
\end{equation}
provided that $\Tr[H_\mu^2] \varepsilon_O \varepsilon_{\sigma_{w}}\in \Omega\left(\frac{1}{\poly(n)}\right)$.
\end{corollary}

Let us now consider the main implication of this result. Corollary~\ref{coro1} provides conditions under which no barren plateaus arise and the trainability of the parameters in the unitaries of the QCNN architecture can be guaranteed.  This is due to the fact that the determination of a cost minimizing direction in the parameter hyperspace requires only a polynomially large precision. This is in contrast to landscapes which exhibit a barren plateau, where an exponentially large precision is needed to navigate through the flat landscape. 

We remark that at the core of the result in Corollary~\ref{coro1} lies the fact that QCNNs are shallow, as they have at most a number of layers $L$  in $\OC(\log(n))$.  Such shallowness naturally arises from the intrinsic structure of the QCNN architecture, where the number of qubits is reduced at each layer. Hence, terms such as $1/50^{L-\ell+1}$ or $\Lambda_{\vec{g}}$ in Theorem~\ref{theo1} can vanish no faster than $\Omega(1/\poly(n))$ if $L$ is in $\OC(\log(n))$.

As indicated in Corollary~\ref{coro1}, the bound in~\eqref{eq:corlower} holds if $\Tr[H_\mu^2] \varepsilon_O \varepsilon_{\sigma_{w}}\in \Omega(1/\poly(n))$. The latter means that  $O$ and the reduced operator $\sigma_w$ need to respectively have a Hilbert-Schmidt distance to the (scaled) identity operator which is not exponentially vanishing. Note that this condition is expected, as extracting information by measuring operators close to the identity is naturally  a hard task. Similarly, attempting to train a quantum circuit on a state close to the identity is also a difficult task.  Here we remark that in the Appendix we relax the condition on $\sigma$ so that we can guarantee that the lower bound for $\Var[\partial_\mu C]$ is at most polynomially vanishing with $n$ if $\sigma_z$ is not exponentially close to $\Tr[\sigma]\id/4$ (as measured by the contraction in $N_1$)  on {\it any} pair of qubits in the backwards-propagated light-cone of $W$.

\subsection{Trainability from pooling}\label{sec:pooling}

As previously discussed, when the unitaries $W_{ij}$ in the convolutional layers form $2$-designs, the action of the operators $I_{ij}$ in the pooling modules can be absorbed by the action of their neighboring $W_{ij}$. Here, we analyze the trainability of the QCNN for the special case when the unitaries $W_{ij}$ act trivially, i.e. when $W_{ij}=\id$ $\forall i,j$.  This can be illustrated by the pooling module in Fig.~(\ref{fig:poolcnn}), in which an input state $\left( \rho^{(0)}\right)^{\otimes n}$ is processed by pooling layers, locally described by the channel
\begin{equation}\label{eq:statepooling}
\rho^{(i+1)}=\text{tr}_{A}\left[ I_{i+1}\left(\rho^{(i)}\right)^{\otimes 2} I_{i+1}\ad \right]\,,
\end{equation}
where 
\begin{equation}\label{eq:oppooling}
I_{i+1}=\ket{+}\bra{+}_{A}\otimes e^{-i\theta_{+}^{(i+1)}Y} + \ket{-}\bra{-}_{A}\otimes e^{-i\theta_{-}^{(i+1)}Y}\,,
\end{equation}
is the controlled unitary corresponding to the pooling. Moreover, here we assume for simplicity that $\rho^{(0)}=\ketbra{0}{0}^{\otimes n}$, where we recall that $n=2^{L}$. 

As explicitly proved in the Appendix, the following theorem holds.
\begin{theorem}\label{theo2}
Consider a CQNN where the convolutional unitaries are $W_{ij}=\id$ for all $i,j$ and where the action of pooling layers can be described by Eqs.~\eqref{eq:statepooling}--\eqref{eq:oppooling}. Then, the expectation value of the magnitude of the cost function partial derivative $\partial C/ \partial\theta^{(k)}=\partial_k C$ with respect to any parameter $\theta^{(k)}$ in a pooling module is
\begin{equation}
\left\langle \vert \partial_k C \vert \right\rangle = {1\over 2n^{\log_{2}(\pi) -1}}\,.
\end{equation}
\end{theorem}

Theorem~\ref{theo2} shows that the pooling-based QCNN does not exhibit a barren plateau as $\langle \vert \partial_k C \vert \rangle$ is polynomially vanishing with the system size. 

\begin{figure}[t]
    \centering
    \includegraphics[width=.95\columnwidth]{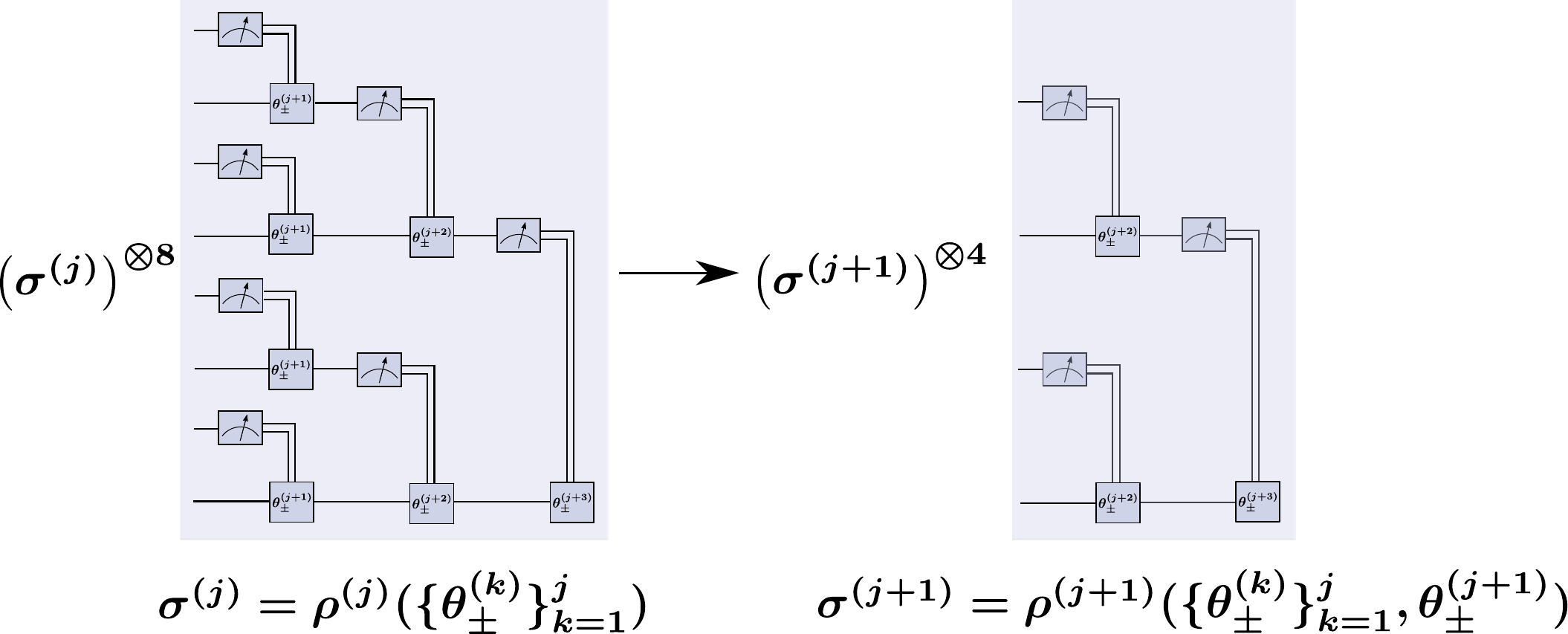}
    \caption{Schematic circuits of a pooling-based QCNN. The left  module belongs  to the layer that acts prior to one which contains the right module. As shown, after each pooling layer, half the qubits are measured.  In this special case, the unitaries in the convolutional and fully connected layers of the QCNN act trivially.}
    \label{fig:poolcnn}
\end{figure}

\subsection{Numerical Verification}\label{sec:numerics}

In this section we present numerical results to analyze the scaling of the cost function partial derivative variance. Specifically, we consider here two cases: when all the unitaries $W_{ij}$ in each convolutional module are independent, and when all the unitaries in the same layer are identical. The latter corresponds to the QCNN architecture as originally introduced in~\cite{cong2019quantum}. As shown below, our numerics indicate that correlating (``corr'') the unitaries $W_{ij}$ (to make them identical) leads to a variance of the cost function partial derivatives which is larger to the one obtained in the uncorrelated (``uncorr'')  case, i.e.,
\begin{equation}\label{eq:corrvsuncorr}
    \Var[\partial_\mu C_{\text{corr}}]\geq \Var[\partial_\mu C_{\text{uncorr}}]\,.
\end{equation}
This suggests  that the lower bound obtained in Theorem~\ref{theo1} can also be used as a lower bound for the correlated case. We remark that in~\cite{volkoff2020large} it was shown that correlating the parameters in an ansatz can lead to larger variances. 

\begin{figure}[t]
    \centering
    \includegraphics[width=.9\columnwidth]{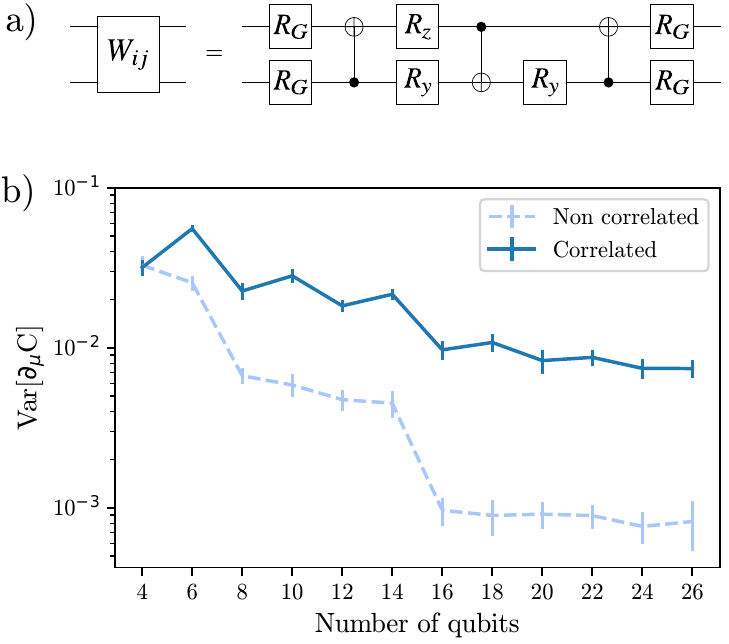}
    \caption{Ansatz and results for the numerical simulations of the QCNN. (a) Ansatz for the two qubit unitary $W_{ij}$. Here $R_G=e^{-i Z \theta_1/2}e^{-i Y \theta_2/2}e^{-i Z \theta_3/2}$ is a general single-qubit rotation parametrized by three angles, and  $X$, $Y$ are Pauli operators. In addition, $R_y$ and $R_z$ denote single-qubit rotations about the $y$ and $z$ axis, respectively. (b) Heuristical scaling of the variance of the cost function partial derivative $\Var[\partial_\mu C]$. The $y$-axis is in a log scale.  The solid (dashed) curve represents results for correlated (uncorrelated) unitaries in the convolutional layers. Error bars represent the standard deviation over 16 repetitions. Here we can see that $\Var[\partial_\mu C]$ is always larger when the unitaries are correlated. As $n$ increases we see that the curves are sub-linear, indicating a sub-exponential scaling for $\Var[\partial_\mu C]$.}
    \label{fig:qcnn-numerics}
\end{figure}

The QCNN in our heuristics is constructed such that each $W_{ij}$ is given by the parametrized circuit of Fig.~\ref{fig:qcnn-numerics}(a). As shown in~\cite{vatan2004optimal}, this decomposition allows us to prepare any two-qubit unitary (up to a global phase). Since we simulate QCNNs with an even number of qubits which was not always a power of $2$, in the pooling layers we trace out qubits so that each convolutional layer always acts on an even number of qubits. Moreover, for simplicity, we consider the case where the training set contains a single state initialized the the all-zero state $\rho=\ketbra{0}{0}^{\otimes n}$, and at the end of the QCNN we measure the operator $O=Z\otimes Z$, where $Z$ denotes the Pauli-$z$ operator. Finally, here the trainable parameter $\theta_\mu$ is in the center unitary acting on the first layer. We remark that the results obtained for this case are representative of other possible choices for the location of $\theta_\mu$.

For each number of qubits $n \in \{4,6,...,26\}$, we simulated $200$ instances of randomly initialized QCNNs and computed the derivative $\partial_\mu C$. All the simulations were performed using TensorFlow Quantum \cite{broughton2020tensorflow}. In Fig.~\ref{fig:qcnn-numerics}(b) we show the results for $\Var[\partial_\mu C]$ as a function of $n$ for the cases when the $W_{ij}$ are correlated and uncorrelated. First, let us remark that the variance for the correlated case is always larger than that of the uncorrelated case as in~\eqref{eq:corrvsuncorr}. Moreover, noting that the $y$-axis is a log scale, we can see that the curves are sub-linear, meaning that the scaling of $\Var[\partial_\mu C]$ is sub-exponential. This result then numerically verifies that QCNNs do not exhibit barren plateaus in their training landscape.

\section{Discussion}\label{sec:discussion}


Investigating the barren plateau phenomenon in parametrized quantum circuits and quantum neural networks is fundamental to understanding their trainability. In the absence of a barren plateau, the determination of a minimizing direction in the cost function landscape does not require an exponentially large precision, meaning that one can always navigate through the landscape by measuring expectation values with a precision that grows at most polynomially with the system size. We remark that polynomial overhead is the standard goal for quantum algorithms, which aim to achieve a speedup over classical algorithms that often scale exponentially. Hence, the absence or existence of a barren plateau can determine whether or not a quantum speedup is achievable.



In this work we analyzed the gradient scaling for the Quantum Convolutional Neural Network (QCNN) architecture. We first introduced a novel method, called the Graph Recursion Integration Method (GRIM), which greatly simplifies the computation of expectation values of quantities that are expressed in terms of Haar-distributed unitaries. The goal of the GRIM is to create  an oriented graph that can be used to evaluate this expectation value.


We then employed the GRIM to derive our main result in Theorem~\ref{theo1}, which provides a general lower bound on the variance of the cost function gradient. In Corollary~\ref{coro1} we derived the asymptotic scaling of this lower bound under standard assumptions and show that it vanishes no faster than polynomially with the system size. This result confirms the absence of barren plateaus in QCNNs and allows us to guarantee the trainability of this architecture when randomly initializing its parameters. Moreover, our work highlights QCNNs as potentially being generically trainable under random initialization, and hence sets them apart from other QNN architectures that are trainable only under certain conditions.  


Since the results in Theorem~\ref{theo1} were derived assuming that the unitaries in the convolutional layers form independent $2$-designs, we additionally considered the case when they are trivial and equal to the identity. In Theorem~\ref{theo2} we show that this pooling-based QCNN, i.e., a QCNN which is composed only of pooling layers, is also trainable when initialized randomly. In addition, we present numerical evidence to support our theoretical results, and we show that correlating the parameters in the convolutional layers (as in the original QCNN architecture), leads to larger gradient variances, meaning that our results can also be useful for that particular case. Similarly, when the two-qubit gates do not form $2$-designs, then one can expect the variance of the gradient to be larger~\cite{holmes2021connecting}. Moreover, let us remark that while we have here considered cost functions which are linear with respect to the input density matrices, as shown in Section~\ref{sec:cost-function}, this type of costs is quite general and encompasses a wide range of applications (such as classification). 


While our work proves the trainability for the QCNN architecture, there are many future research directions that can be pursued to generalize our results. First, one can analyze more general cost functions than the ones considered here. For instance, one could consider mean-square cost functions, as the ones employed in regression problems, where $C$ is a quadratic function on the expectation value $\Tr[V(\thv)\sigma C\ad(\thv)O]$.  In addition, due to close connection between the QCNN architecture and the Multi-scale Entanglement Renormalization Ansatz (MERA)~\cite{cong2019quantum}, it would be interesting to analyze how our results and methods can be employed to derive trainability results for MERA frameworks. Finally, we believe the  GRIM will be useful in analyzing the scaling of the cost function partial derivative variance for other VQA and QNN architectures. 

{\it Note added.} After completion of this work, \cite{zhang2020toward} was posted where the authors analyze the existence of barren plateaus in QNN with a Tree Tensor structure, concluding that this architecture also does not exhibit barren plateaus. 

\section*{Acknowledgements}

We thank Lukasz Cincio for helpful discussions. AP and SW were supported by the U.S. Department of Energy (DOE) through a quantum computing program sponsored by the Los Alamos National Laboratory (LANL) Information Science \& Technology Institute. SW also acknowledges support from the Samsung GRP grant. MC acknowledges initial support from the Center for Nonlinear Studies at Los Alamos National Laboratory (LANL). TV (ATS) was supported by the Laboratory Directed Research and Development (LDRD) program of LANL under project number 20200677PRD1 (20190065DR, resp.). ATS and PJC acknowledge initial support from the LANL ASC Beyond Moore's Law project. This work was supported by the U.S. DOE, Office of Science, Office of Advanced Scientific Computing Research, under the Accelerated Research in Quantum Computing (ARQC) program.

\bibliography{ref.bib}

\setcounter{section}{0}
\onecolumngrid

\setcounter{section}{0}

\renewcommand{\thesection}{\arabic{section}}

\titleformat{\section}{\large\bfseries}{}{0pt}{ Appendix \thesection:\quad}

\section{Preliminaries}

In this section, we present preliminary notation and results needed to derive our main Theorems. 

\subsection{Haar distributed unitaries and integration over the unitary group}

Let us here recall the definition of a $t$-design. Let $\WC=\{W_y\}_{y\in Y}$ be a finite set of size $|Y|$ of $d$-dimensional unitaries $W_y$, and let $P_{t,t}(W)$ be a polynomial of degree at most $t$ in the matrix elements of $W$, and at most of degree $t$ in those of $W\ad$. We say that $\WC$ is a $t$-design if for every polynomial $P_{t,t}(W)$ we have
\begin{equation}\label{eq:design}
    \frac{1}{|Y|}\sum_{y\in Y} P_{t,t}(W_y)=\int d\mu(W) P_{t,t}(W)\,,
\end{equation}
where the integral is taken with respect to the Haar measure over the unitary group.

Then, to compute expectation values as in~\eqref{eq:design} it is useful to recall formulas which allow for  symbolical integration with respect to the Haar measure on the unitary group $\mathcal{U}(d)$ of degree $d$. Specifically, following the elementwise formula of the Weingarten calculus~\cite{collins2006integration,puchala2017symbolic} we have that the following formulas are valid for the first two moments: 
\begin{equation}
\begin{aligned}
\!\!\int d\mu(W)w_{\vec{i}_1\vec{j}_1}w_{\vec{i}_2\vec{j}_2}^* &= \frac{\delta_{\vec{i}_1\vec{i}_2}\delta_{\vec{j}_1\vec{j}_2}}{d}\! \\
\!\!\int d\mu(W)w_{\vec{i}_1\vec{j}_1}w_{\vec{i}_2\vec{j}_2}w_{\vec{i}_1'\vec{j}_1'}^{*}w_{\vec{i}_2'\vec{j}_2'}^{*}&=\frac{1}{d^2-1}\left( \! \Delta_1-\frac{\Delta_2}{d}\! \right)\,,\label{eq:Haarmoment}
\end{aligned}
\end{equation}
where $w_{\vec{i}\vec{j}}$ are the matrix elements of the unitary $W \in \mathcal{U}(d)$, and
\begin{equation}
\begin{aligned}
\Delta_1&=\delta_{\vec{i}_1\vec{i}_1'}\delta_{\vec{i}_2\vec{i}_2'}\delta_{\vec{j}_1\vec{j}_1'}\delta_{\vec{j}_2\vec{j}_2'}+\delta_{\vec{i}_1\vec{i}_2'}\delta_{\vec{i}_2\vec{i}_1'}\delta_{\vec{j}_1\vec{j}_2'}\delta_{\vec{j}_2\vec{j}_1'}\,,\\
\Delta_2&=\delta_{\vec{i}_1\vec{i}_1'}\delta_{\vec{i}_2\vec{i}_2'}\delta_{\vec{j}_1\vec{j}_2'}\delta_{\vec{j}_2\vec{j}_1'}+\delta_{\vec{i}_1\vec{i}_2'}\delta_{\vec{i}_2\vec{i}_1'}\delta_{\vec{j}_1\vec{j}_1'}\delta_{\vec{j}_2\vec{j}_2'}\,.    \label{eq:Haar2moment2}
\end{aligned}
\end{equation} 

Here we remark that the integrations in our calculations were performed by employing the Random Tensor Network Integrator (RTNI) package for symbolic integration of Haar distributed tensors~\cite{Fukuda_2019}. 

\subsection{Useful results}

\begin{lemma} \label{lemma:trace-identity}
Let $\HC=\HC_1 \otimes \HC_2 \otimes \HC_3$ be a tripartite Hilbert space, and let $\{\ket{\vec{p}}\}$ be an arbitrary basis of $\HC_3$. Then for any pair of linear operators $A,B: \HC \rightarrow \HC$, the following identities hold:

\begin{align} \label{eq:lemma1-eq1}
   \sum_{\vec{p},\vec{q}}\Tr_2\left[\Tr_{13}[(\id_{12}\otimes \ketbra{\vec{p}}{\vec{q}}_3)A\right]\Tr_{13}\left[(\id_{12}\otimes \ketbra{\vec{q}}{\vec{p}}_3)B]\right]
   =\Tr_{23}[\Tr_{1}[A]\Tr_{1}[B]]\,,
\end{align}
\vspace{-10pt}
\begin{align}
   \sum_{\vec{p},\vec{q}}\Tr_{123}\left[(\id_{12}\otimes \ketbra{\vec{p}}{\vec{q}}_3)A\right]\Tr_{123}[(\id_{12}\otimes \ketbra{\vec{q}}{\vec{p}}_3)B]]
   =\Tr_{3}[\Tr_{12}[A]\Tr_{12}[B]]\,.
\end{align}
\end{lemma}
For the proof of Lemma~\ref{lemma:trace-identity}, see~\cite{cerezo2020cost}.

\begin{lemma} \label{lemma:trace-dhs}
    Let $\HC=\HC_1 \otimes \HC_2$ be a bipartite Hilbert space with $\dim \HC_2=d$, and let $H$ be a Hermitian operator on $\HC$, we have:
    \begin{align}
        \Tr\left[ H^2 \right] - \Tr\left[ \frac{\Tr_2\left[H \right]^2}{d} \right]
        =
        D_{HS}\left(H,\Tr_2\left[H \right] \otimes \frac{\id}{d} \right)\,,
    \end{align}
    where $D_{HS}(H_1,H_2)=\Tr[(H_1 - H_2)^2]=||H_1 - H_2||_2^2$ is the Hilbert-Schmidt distance. In particular, if $\HC=\HC_2$
    \begin{align} \label{eq:delta-as-distance}
        \Tr\left[ H^2 \right] - \frac{1}{d} \Tr[H]^2
        =
        D_{HS}\left(H, \frac{\id}{d} \Tr[H] \right)\,.
    \end{align}
\end{lemma}
\begin{proof}
Consider the following chain of equalities: 
\begin{align}
    D_{HS}\left(H,\Tr_2[H]\otimes \frac{\id}{d} \right)
    &=
    \Tr\left[ \left( H - \Tr_2[H] \otimes \frac{\id}{d} \right)^2 \right] \\
    &=
    \Tr\left[ H^2 \right] + \Tr\left[\Tr_2[H]^2 \otimes \frac{\id}{d^2} \right] - 2 \Tr\left[ H \left( \Tr_2[H] \otimes \frac{\id}{d} \right) \right] \\
    &=
    \Tr\left[ H^2 \right] + \frac{1}{d} \Tr\left[\Tr_2[H]^2 \right] - 2 \Tr\left[ H \left( \Tr_2[H] \otimes \frac{\id}{d} \right) \right] \\
    &=
    \Tr\left[ H^2 \right] + \frac{1}{d} \Tr\left[\Tr_2[H]^2 \right] - \frac{2}{d} \Tr\left[ \Tr_2[H]^2 \right] \\
    &=
    \Tr\left[ H^2 \right] - \frac{\Tr\left[\Tr_2[H]^2 \right]}{d}\,.
\end{align}
\end{proof}

\begin{lemma} \label{lemma:s1-s2}
Consider the operator $\Delta\Psi_{\vec{p}\vec{q}}^{\vec{p'}\vec{q'}}$ defined in the main text, which we recall here for convenience:
\begin{equation}\label{eq:SMdeltapsi}
    \Delta\Psi_{\vec{p}\vec{q}}^{\vec{p'}\vec{q'}} =  \Tr[\Psi_{\vec{p}\vec{q}}\Psi_{\vec{p'}\vec{q'}}]-\frac{\Tr[\Psi_{\vec{p}\vec{q}}]\Tr[\Psi_{\vec{p'}\vec{q'}}]}{d}\,,
\end{equation}
where $\vec{p}$, $\vec{q}$, $\vec{p'}$, $\vec{q'}$ are vectors of length $2^{n}-d$. Let $\HC$ be a Hilbert space of $n$ qubits and let $\HC_w$ be a  $d$-dimensional subsystem, such that $\HC_w\otimes\HC_{\overline{w}}=\HC$. Then, let $S$ and $\Sb$ be two disjoint sets of qubits not in $\HC_w$ such that  $\HC_S\otimes\HC_{\overline{S}}=\HC_{\overline{w}}$. The following equality holds
\begin{align}
    \sum_{\substack{\vec{p}\vec{q} \\\vec{p'}\vec{q'}}}
    (\delta_{\vec{p}\vec{q'}})_{\Sb} (\delta_{\vec{q}\vec{p'}})_{\Sb} (\delta_{\vec{p}\vec{q}})_{S} (\delta_{\vec{p'}\vec{q'}})_{S}
    \Delta \Psi_{\vec{p}\vec{q}}^{\vec{p'}\vec{q'}}
    =
    D_{HS}\left(\sigmat_{\Sb, w}, \sigmat_{\Sb} \otimes \frac{\id}{d} \right)
\end{align}
where $\sigmat_{\Sb, w} = \Tr_{S} [\sigma_L]$ and $\sigma_L=V_L \sigma V_L\ad$. Here $\Tr_{S}[\cdot]$ indicates the partial trace over $\HC_S$.
\end{lemma}
\begin{proof}
Let us consider the first term in~\eqref{eq:SMdeltapsi}. We have:
\begin{align}
    & \sum_{\substack{\vec{p}\vec{q} \\\vec{p'}\vec{q'}}}
    (\delta_{\vec{p}\vec{q'}})_{\Sb} (\delta_{\vec{q}\vec{p'}})_{\Sb} (\delta_{\vec{p}\vec{q}})_{S} (\delta_{\vec{p'}\vec{q'}})_{S}
    \Tr\left[ \Psi_{\vec{p}\vec{q}} \Psi_{\vec{p'}\vec{q'}} \right]
    \\
    &=
    \sum_{\substack{\vec{p}\vec{q} \\\vec{p'}\vec{q'}}}
    (\delta_{\vec{p}\vec{q'}})_{\Sb} (\delta_{\vec{q}\vec{p'}})_{\Sb} (\delta_{\vec{p}\vec{q}})_{S} (\delta_{\vec{p'}\vec{q'}})_{S}
    \Tr\left[ 
        \Tr_{\wbar}\left[ (\ketbra{\vec{p}}{\vec{q}}_{\Sb} \otimes \ketbra{\vec{p}}{\vec{q}}_{S} \otimes \id_d) \sigma_L \right] 
        \Tr_{\wbar}\left[ (\ketbra{\vec{p'}}{\vec{q'}}_{\Sb} \otimes \ketbra{\vec{p'}}{\vec{q'}}_{S} \otimes \id_{w}) \sigma_L \right] 
    \right] \\
    &=
    \sum_{\vec{p}\vec{q} \in \Sb} \sum_{\vec{p}\vec{p'} \in S}
    \Tr\left[ 
        \Tr_{\wbar}\left[ \left(\ketbra{\vec{p}}{\vec{q}}_{\Sb} \otimes \ketbra{\vec{p}}{\vec{p}}_{S} \otimes \id_{w}\right) \sigma_L  \right] 
        \Tr_{\wbar}\left[ \left(\ketbra{\vec{q}}{\vec{p}}_{\Sb} \otimes \ketbra{\vec{p'}}{\vec{p'}}_{S} \otimes \id_{w}\right) \sigma_L  \right] 
    \right] \\
    &=
    \sum_{\vec{p}\vec{q} \in \Sb}
    \Tr\left[ 
        \Tr_{\wbar}\left[ \left(\ketbra{\vec{p}}{\vec{q}}_{\Sb} \otimes \sum_{\vec{p}} \ketbra{\vec{p}}{\vec{p}}_{S} \otimes \id_{w}\right) \sigma_L \right] 
        \Tr_{\wbar}\left[ \left(\ketbra{\vec{q}}{\vec{p}}_{\Sb} \otimes \sum_{\vec{p'}} \ketbra{\vec{p'}}{\vec{p'}}_{S} \otimes \id_{w}\right) \sigma_L  \right] 
    \right] \\
    &=
    \sum_{\vec{p}\vec{q} \in \Sb}
    \Tr\left[ 
        \Tr_{\wbar}\left[ (\ketbra{\vec{p}}{\vec{q}}_{\Sb} \otimes \id_{w, S}) \sigma_L  \right] 
        \Tr_{\wbar}\left[ (\ketbra{\vec{q}}{\vec{p}}_{\Sb} \otimes \id_{w, S}) \sigma_L  \right] 
    \right] \\
    &=
    \sum_{\vec{p}\vec{q} \in \Sb}
    \Tr\left[ 
        \Tr_{\wbar}\left[ (\ketbra{\vec{p}}{\vec{q}}_{\Sb} \otimes \id_{w, S}) \sigma_L  \right] 
        \Tr_{\wbar}\left[ (\ketbra{\vec{q}}{\vec{p}}_{\Sb} \otimes \id_{w, S}) \sigma_L  \right] 
    \right] \\
    &=
    \Tr\left[ 
        \Tr_{S}[\sigma_L]^2
    \right] \\
    &=
    \Tr\left[ \sigmat_{\Sb, w}^2 \right] \label{eq:lem3res1}\,,
\end{align}
where the last line comes from Lemma \ref{lemma:trace-identity}.

Similarly, for the second term in~\eqref{eq:SMdeltapsi} we have
\begin{align}
    \sum_{\substack{\vec{p}\vec{q} \\\vec{p'}\vec{q'}}}
    (\delta_{\vec{p}\vec{q'}})_{\Sb} (\delta_{\vec{q}\vec{p'}})_{\Sb} (\delta_{\vec{p}\vec{q}})_{S} (\delta_{\vec{p'}\vec{q'}})_{S}
    \Tr\left[ \Psi_{\vec{p}\vec{q}} \right] \Tr \left[ \Psi_{\vec{p'}\vec{q'}} \right]
    &=
    \Tr\left[ 
        \Tr_{w}[\sigmat_{\Sb, w}]^2
    \right]\,.\label{eq:lem3res2} \\
\end{align}

Therefore, combining the results from~\eqref{eq:lem3res1} and~\eqref{eq:lem3res2} we find 
\begin{align}
    \sum_{\substack{\vec{p}\vec{q} \\\vec{p'}\vec{q'}}}
    (\delta_{\vec{p}\vec{q'}})_{\Sb} (\delta_{\vec{q}\vec{p'}})_{\Sb} (\delta_{\vec{p}\vec{q}})_{S} (\delta_{\vec{p'}\vec{q'}})_{S}
    \Delta \Psi_{\vec{p}\vec{q}}^{\vec{p'}\vec{q'}}
    &=
    \Tr \left[
        \Tr_w \left[\sigmat_{\Sb, w}^2\right]
        -
        \frac{
            \Tr_w \left[\sigmat_{\Sb, w}\right]^2
        }{d}
    \right] \\
    &= D_{HS}\left(\sigmat_{\Sb, w}, \sigmat_{\Sb} \otimes \frac{\id}{d} \right)\,,
\end{align}
where in the last equation we invoked Lemma~\ref{lemma:trace-dhs}.
\end{proof}

\begin{lemma}[Monotonicity of the trace distance] \label{lemma:monotonicity}
Let $\HC$ be a  tensor product Hilbert space $\HC=\HC_A\otimes\HC_B$, and let $R_{AB}$ and $S_{AB}$ be two operators on $\HC$. Then, the following inequality holds 
\begin{equation}
    D_T(R_{AB},S_{AB}) \geq D_T(R_{B},S_{B})\,,
\end{equation}
where $R_{B} = \Tr_A (R_{AB})$ and $S_{B} = \Tr_A (S_{AB})$. Here
$D_T(R_{B},S_{B})=\frac{1}{2}||R_{B} - S_{B}||_1$ is the trace distance, with $||\cdot||_1$  the Schatten $1$-norm.

\begin{remark}
    The monotonicity of the trace distance is usually proven when $R_{AB}$ and $S_{AB}$ are density matrices. The proof we present here is valid for any two matrices.
\end{remark}

\end{lemma}
\begin{proof}
By duality of Schatten norms~\cite{horn2012matrix}, we have:
\begin{align}\label{eq:duality}
    \begin{split}
        ||R_{B} - S_{B}||_1
        &=
        \max_{\Lambda_B \text{ s.t. }||\Lambda_B||_{\infty} = 1} \Tr[\Lambda_B (R_{B} - S_{B})]\,.
    \end{split}
\end{align}
Therefore, there exists a matrix $\Lambda_B^*$ with $||\Lambda_B||_{\infty} = 1$ such that the maximum in~\eqref{eq:duality}  is reached, i.e.,
\begin{align}
    \begin{split}
        ||R_{B} - S_{B}||_1
        &=
        \Tr[\Lambda_B^* (R_{B} - S_{B})] \\
        &=
        \Tr[(\id_A \otimes \Lambda_B^*) (R_{AB} - S_{AB})] \\
        &\leq
        \max_{\Lambda_{AB} \text{ s.t. }||\Lambda_{AB}||_{\infty} = 1}\Tr[\Lambda_{AB} (R_{AB} - S_{AB})] \\
        &=
        ||R_{AB} - S_{AB}||_1\,.
    \end{split}
\end{align}
The inequality follows from~\eqref{eq:duality} and the fact that $||\id_A \otimes \Lambda_B^*||_{\infty}=1$, while the last equality is simply another application of the Schatten norm duality.

We remark that an alternative proof, for the case where $R_{AB}$ and $S_{AB}$ are Hermitian, proceeds as follows. Let us begin by noting that one can write $R_{AB} - S_{AB} = P_{AB} - Q_{AB}$ where $P_{AB} \geq 0$ and $Q_{AB} \geq 0$, and $P_{AB}$ is orthogonal to $Q_{AB}$. The orthogonality condition gives $|P_{AB} - Q_{AB}| = P_{AB} + Q_{AB}$. One then has
\begin{align}
    D_T(R_{AB},S_{AB}) &= \frac{1}{2}\Tr |P_{AB} - Q_{AB}|\\ 
    &= \frac{1}{2}(\Tr P_{AB} +\Tr Q_{AB})\\
    &= \frac{1}{2}(\Tr P_{B} +\Tr Q_{B}) \\
    &= \frac{1}{2}( \|P_{B} \|_1 +  \| Q_{B} \|_1 )\\
    &\geq \frac{1}{2}( \|P_{B} -  Q_{B} \|_1 )\\
    &= \frac{1}{2}( \|R_{B} -  S_{B} \|_1 )\,,
\end{align}
where $P_{B} = \Tr_A (P_{AB})$ and $Q_{B} = \Tr_A (Q_{AB})$. Note that the inequality in this proof follows from the triangle inequality for matrix norms.
\end{proof}

\begin{lemma} \label{lemma:dhs-inequality}
Let $\HC$ be a  tensor product Hilbert space $\HC=\HC_A\otimes\HC_B$, with $\dim{H_A}=d_A$ and $\dim{H_B}=d_B$, and let $R_{AB}$ be an operator on $\HC$. Then, the following inequality holds: 
    \begin{align}   
        D_{HS}\left(R_{AB},R_{A}\otimes \frac{\id_{B}}{d_B}\right)
        \geq \frac{D_{HS}\left(R_{B},\frac{\id}{d_B}\right)}{d_B d_{A}}.
    \end{align}
\end{lemma}
\begin{proof}
The equivalence of matrix norms gives us the following inequality between the Frobenius norm $\norm{M}_2$ and the trace norm $\norm{M}_1$ of a matrix $M$ with dimension $d$ \cite{coles2019strong}:
\begin{align} \label{eq:matrix-norm-equivalence}
    \norm{M}_2 \leq \norm{M}_1 \leq \sqrt{d} \norm{M}_2 
\end{align}
The first inequality comes from the monotonicity of the Schatten norms \cite{raissouli2010various, horn2012matrix} and the second one from the Cauchy-Schwartz inequality. Since $D_{HS}(R,S)=\norm{R - S}_2^2$ and $D_{T}(R,S)=\frac{1}{2}\norm{R - S}_1$, it follows from Eq. (\ref{eq:matrix-norm-equivalence}) that:
\begin{align} \label{eq:equivalence-distance}
    \frac{1}{4} D_{HS}(R,S) \leq D_T(R,S)^2 \leq \frac{d}{4} D_{HS}(R,S)
\end{align}
Therefore, in our case:
\begin{align}
    D_{HS}\left(R_{AB},R_{A}\otimes \frac{\id_{B}}{d_B}\right)
    &\geq \frac{4D_{T}\left(R_{AB},R_{A}\otimes \frac{\id_{B}}{d_B}\right)^2}{d_B d_{A}}\label{eq:ineq-TD0}\\
    &\geq \frac{4 D_{T}\left(R_{B},\frac{\id}{d_B}\right)^2}{d_B d_{A}}\nonumber\\
    &\geq \frac{D_{HS}\left(R_{B},\frac{\id}{d_B}\right)}{d_B d_{A}}\,.\label{eq:ineq-TD}
\end{align}
where the first and last inequalities come from the equivalence of matrix norms given by Eq.~\eqref{eq:equivalence-distance}, while the second inequality follows from the monotonicity of the Trace distance as stated in Lemma~\ref{lemma:monotonicity}.
\end{proof}

\section{Proof of Theorem 1}

In this section we present a proof of Theorem~\ref{theo1}, which establishes a lower bound for the variance of the cost function partial derivative in terms of the coefficients of a graph obtained through the GRIM. Moreover, in this section we also present a proof for Proposition~\ref{prop1}. Let us  restate the theorem here for convenience:

\setcounter{theorem}{0}
\begin{theorem}\label{theo1SM}
Consider the QCNN cost function of Eq.~\eqref{eq:costreduced} and let $W$ be a unitary in the first sub-layer of the $\ell$-th layer of $V(\thv)$. Moreover, let us assume  that the unitaries $W_{ij}$, $W_A$ and $W_B$ in the convolutional and fully-connected layers form independent $2$-designs. Then, a lower bound of the variance of the cost function partial derivative, with respect to a parameter $\theta_\mu$ in  unitary $W$, can be computed using the GRIM from the graph $\GC_w$ as
\begin{align}
    \Var[\partial_\mu C]\geq F_n(L,\ell)\,,
\end{align}
with
\begin{equation}\label{eq:functionlowerSM}
    F_n(L,\ell)=\frac{1}{9}\frac{\Tr[H_\mu^2] \varepsilon_O \varepsilon_{\sigma_w} }{50^{L-\ell+1}} \sum_{\vec{g}\in P_\ell(\GC_w)} \Lambda_{\vec{g}}.
\end{equation}
Here $\sigma_w=\Tr_{\overline{w}}[\sigma]$ is the reduced operator $\sigma$ in the subsystem of the qubits that $W$ acts on. Moreover, $\varepsilon_O$, $P_\ell(\GC_w)$, $\vec{g}$ and $\Lambda_{\vec{g}}$ are defined in Proposition~\ref{prop1}.
\end{theorem}

\subsection{Variance of the cost function partial derivative}

Let us assume here that  the block $W$ is in the $\ell$-th layer of the QCNN. As discussed in the main text, we can decompose $V$ as 
\begin{align*}
    V=V_R W V_L\,,
\end{align*}
where $V_R$ contains all the gates in the forward light-cone $\LC_F$ of $W$, and where $V_L$ contains all remaining gates. Moreover, $W$ can also be decomposed as
\begin{align*}
    W=W_A W_B\,,
\end{align*}
where
\begin{align*}
 W_A=\prod_{\eta\leq\mu} e^{-i \theta_{\eta}H_\eta}W_\eta\,, \quad   W_B=\prod_{\eta>\mu} e^{-i \theta_{\eta}H_\eta}W_\eta\,.
\end{align*}
As shown in \cite{cerezo2020cost}, the partial derivative of $C(\thv)$ with respect to $\theta_\mu$ can be written as
\begin{align*}
    \partial_\mu C \!=\! \Tr\left[W_A V_L\sigma V_L\ad W_A\ad [H_\mu,W_B\ad V_R\ad \widetilde{O} V_RW_B]\right],
\end{align*}
and if either $W_A$ or $W_B$ form a $1$-design, we have
\begin{align*}
    \avg{\partial_\mu C}=0\,,
\end{align*}
and
\begin{align}\label{eq:VarcSM}
 \Var[\partial_\mu C]
 =
 \frac{2\Tr[H_\mu^2]}{225}
 \sum_{\substack{\vec{p}\vec{q}\\\vec{p'}\vec{q'}}}\avg{\Delta\Omega_{\vec{q}\vec{p}}^{\vec{q'}\vec{p'}}}_{V_R}\avg{\Delta\Psi_{\vec{p}\vec{q}}^{\vec{p'}\vec{q'}}}_{V_L}\,.
\end{align}
Here the summation runs over all bitstrings $\vec{p}$, $\vec{q}$, $\vec{p'}$, $\vec{q'}$ of length $2^{n-2}$. In addition, we defined 
\begin{align*}
        \Delta\Omega_{\vec{q}\vec{p}}^{\vec{q'}\vec{p'}} &= \Tr[ \Omega_{\vec{q}\vec{p}}\Omega_{\vec{q'}\vec{p'}}] -\frac{\Tr[ \Omega_{\vec{q}\vec{p}}]\Tr[\Omega_{\vec{q'}\vec{p'}}]}{4}\,, \\
    \Delta\Psi_{\vec{p}\vec{q}}^{\vec{p'}\vec{q'}} &=  \Tr[\Psi_{\vec{p}\vec{q}}\Psi_{\vec{p'}\vec{q'}}]-\frac{\Tr[\Psi_{\vec{p}\vec{q}}]\Tr[\Psi_{\vec{p'}\vec{q'}}]}{4}\,,
\end{align*}
 where  $\Omega_{\vec{q}\vec{p}}$ and $\Psi_{\vec{q}\vec{p}}$ are operators on $\HC_{w}$ defined as 
\begin{align*}
    \Omega_{\vec{q}\vec{p}}&=\Tr_{\wbar}\left[(\ketbra{\vec{p}}{\vec{q}}\otimes\id_{w})V_R\ad \widetilde{O}V_R\right]\,,
    \\
    \Psi_{\vec{p}\vec{q}}&=\Tr_{\wbar}\left[(\ketbra{\vec{q}}{\vec{p}}\otimes\id_{w})V_L \sigma V_L\ad\right]\,.
\end{align*}
Here $\Tr_{\wbar}$ indicates the trace over all qubits not in $\HC_{w}$, with $\HC_w$ being the Hilbert space associated with the qubits that that $W$ acts on.

In what follows, we employ the GRIM to compute the contributions of the expectation values $\langle\Delta\Omega_{\vec{q}\vec{p}}^{\vec{q'}\vec{p'}}\rangle_{V_R}$ and $\langle\Delta\Psi_{\vec{p}\vec{q}}^{\vec{p'}\vec{q'}}\rangle_{V_L}$ to the variance of the cost function partial derivative. 

\subsection{Integration over the unitaries in $V_R$ via the GRIM}
\label{sec:appendix-grim}

\subsubsection{Grouping the unitaries in modules}

\begin{figure}[t]
    \centering
    \includegraphics[width=.95\columnwidth]{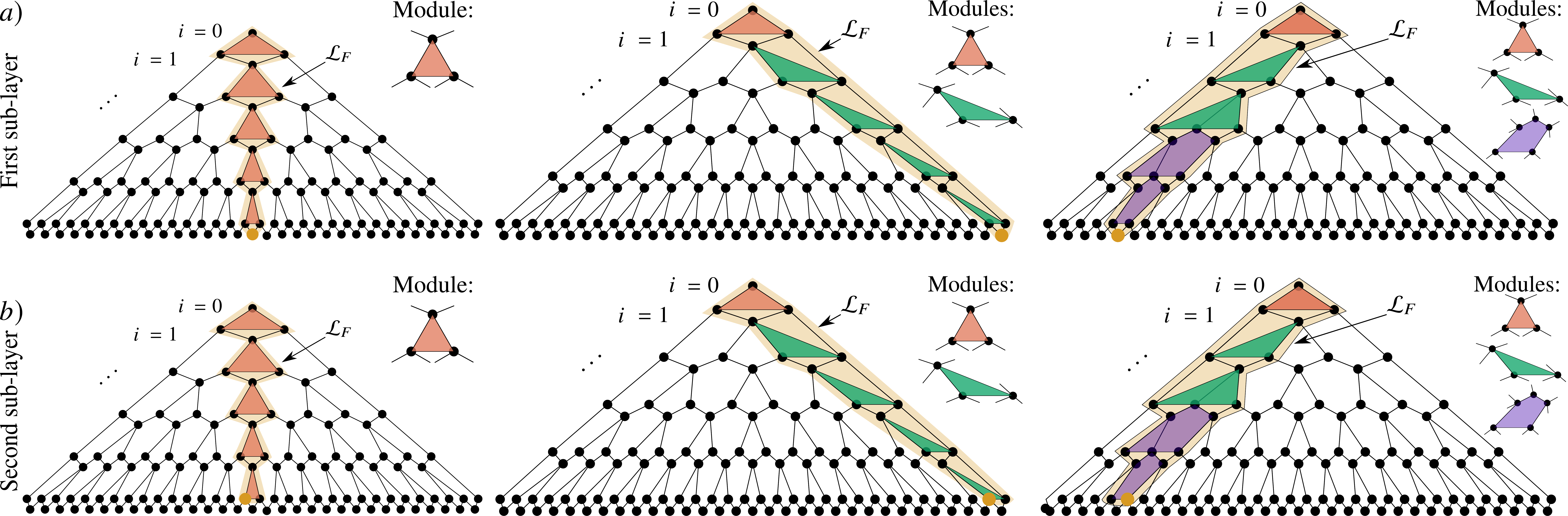}
    \caption{Examples of the first step of the GRIM. Here we can cover all of the forward light-cone $\LC_F$ using three basic modules repeated $\ell$ times. Panels (a) and (b) respectively correspond to the cases when the unitary $W$ is in first sub-layer and second sub-layer of the $\ell$-th QCNN layer. In all cases,  $W$ is indicated by a larger orange circle and the forward light-cone is shaded in yellow.  For the panels in (a) $W$ is not included in any module, while in (b) $W$ belongs to the first module $M$.   }
    \label{fig:all-module-2}
\end{figure}

Let us now focus on the term $\langle\Delta\Omega_{\vec{q}\vec{p}}^{\vec{q'}\vec{p'}}\rangle_{V_R}$. As discussed in the main text we here employ the GRIM, whose first step is to group  the unitaries in $\LC_F$ into modules $M$. Moreover, as show in Fig.~\ref{fig:all-module-2}, for the QCNN architecture one only needs three basic modules: a \textit{center module} $M_{\CC}$, a \textit{middle module} $M_{\MC}$ and an \textit{edge module} $M_{\EC}$. Note that different choices of $W$ lead to different modules being employed. Surprisingly, for any $W$ we will have one of the following three cases:
\begin{itemize}
    \item Case 1: The light-cone $\LC_F$ can be covered by $\ell$ modules $M_{\CC}$.
    \item Case 2: The light-cone $\LC_F$ can be covered by $(\ell-1)$ modules $M_{\EC}$ and a module $M_{\CC}$.
    \item Case 3: The light-cone $\LC_F$ can be covered by $\ell'$ modules $M_{\MC}$, $\ell''$ modules $M_{\EC}$ and a module $M_{\CC}$ such that $\ell'+\ell''+1=\ell$.
\end{itemize}

Let us make the following important remarks. First, note that the {\it final} module is always $M_{\CC}$. Second, let us recall that $W$  can be in the first or in the  second sub-layer. As indicated in Fig.~\ref{fig:all-module-2}(a), if $W$ is in the first sub-layer, then $W$ does not belong to any module. However, if $W$ is in the second sub-layer, then it is always one of the unitaries in a given module (see Fig.~\ref{fig:all-module-2}(b)). This can be alternatively thought as introducing additional basic {\it initial} modules corresponding to taking $M_{\CC}$, $M_{\MC}$ and $M_{\EC}$ and removing a unitary from the second sub-layer.

More formally, the first step of the GRIM implies that $V_R$ can always be decomposed into a set of $\ell$ modules:
\begin{align}
    V_R = V_R^{(\ell)} V_R^{(\ell-1)} \cdots V_R^{(1)}\,,
\end{align}
where $V_R^{(i)}\in\{M_{\CC},M_{\EC},M_{\MC}\}$ for $2 \leq i \leq \ell-1$, while $V_R^{(1)}$ and $V_R^{(\ell)}=M_{\CC}$ are initial and final modules. To integrate the initial module $V_R^{(1)}$ it is convenient to rewrite $\Omega_{\vec{q}\vec{p}}$ as
\begin{align}
    \Omega_{\vec{q}\vec{p}}
    &=
    \Tr_{\wbar}\left[(\ketbra{\vec{p}}{\vec{q}}\otimes\id_{w}) \left(V_R^{(1)}\right)\ad \cdots \left(V_R^{(\ell-1)}\right)\ad \left(V_R^{(\ell)}\right)\ad O V_R^{(\ell)} V_R^{(\ell-1)} \cdots V_R^{(1)} \right] \\ \label{eq:omega-pq-first-module}
    &=
    \Tr_{\wbar}\left[(\ketbra{\vec{p}}{\vec{q}}\otimes\id_{w}) \left(V_R^{(1)}\right)\ad \Ot_{\ell-1} V_R^{(1)} \right]\,,
\end{align}
where we defined
\begin{align}
    \Ot_{i}
    =
    \left(V_R^{(\ell-i+1)}\right)\ad \cdots \left(V_R^{(\ell-1)}\right)\ad \left(V_R^{(\ell)}\right)\ad
    O 
    V_R^{(\ell)} V_R^{(\ell-1)} \cdots V_R^{(\ell-i+1)}.
\end{align}
Finally, let $S_{\wbar}$ be the set of qubits in $\HC_{\overline{w}}$, that is, the set of qubits that $W$ acts trivially on. We decompose $S_{\wbar}$ as the disjoint union
\begin{align}
    S_{\wbar}  = \LC_1 \cup ... \LC_\ell \cup \LCb \,,
\end{align}
where $\LC_k$ represents all the qubits in $S_{\wbar}$ that $V_R^{(k)}$ acts non-trivially on. Moreover, we define $\LCb=S_{\wbar} \setminus \left( \LC_1 \cup ... \cup \LC_\ell \right)$ as the set of all the qubits that are not in the forward light-cone of $W$. Using this decomposition we can rewrite the partial trace of Eq.~\eqref{eq:omega-pq-first-module} as
\begin{align}
    \Omega_{\vec{q}\vec{p}}
    &=
    \Tr_{\LC_1,...,\LC_\ell,\LCb}\left[(\ketbra{\vec{p}}{\vec{q}}\otimes\id_{w}) \left(V_R^{(1)}\right)\ad \Ot_{\ell-1} V_R^{(1)} \right]\,.
\end{align}

In Fig.~\ref{fig:omega-pq} we show the tensor network representation of $\Omega_{\vec{q}\vec{p}}$ and of  $\Delta\Omega_{\vec{q}\vec{p}}^{\vec{q'}\vec{p'}}$ for the case when $V_R^{(1)}=M_{\MC}$.  In the next section we show how to generate the graph $\GC_w$ via the GRIM.

\begin{figure}[t]
    \centering
    \includegraphics[width=.8\columnwidth]{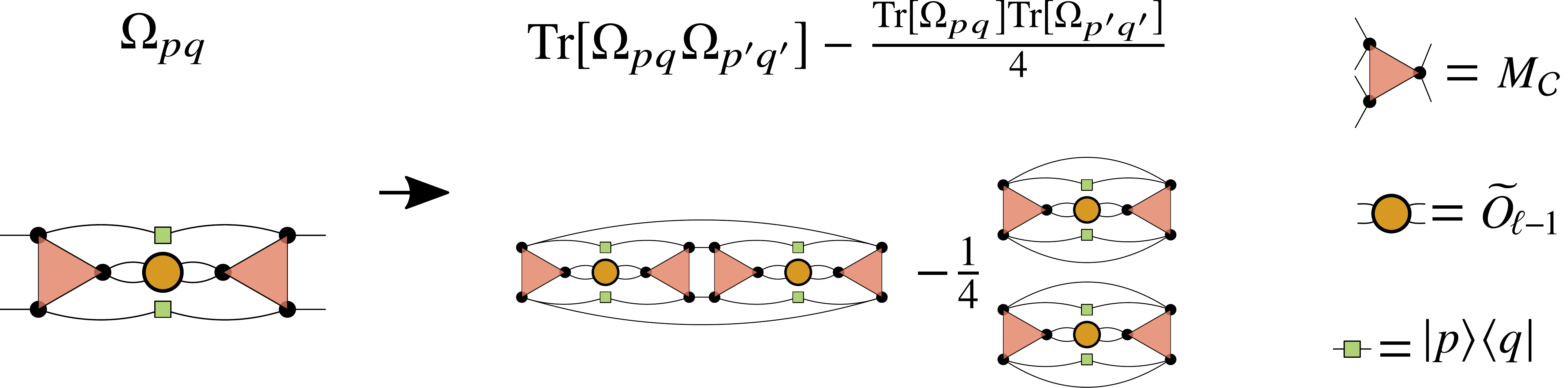}
    \caption{Tensor network representation of  $\Omega_{\vec{q}\vec{p}}$ and $\Delta\Omega_{\vec{q}\vec{p}}^{\vec{q'}\vec{p'}}$. Here we consider the case when $W$ is in the middle of the $\ell$-th layer.  }
    \label{fig:omega-pq}
\end{figure}

\subsubsection{Basic notation}

Here we introduce the basic notation employed to form the GRIM graph. For simplicity, let us consider  the tensor network of $\Delta\Omega_{\vec{q}\vec{p}}^{\vec{q'}\vec{p'}}$ as shown in Fig.~\ref{fig:omega-pq}, and let us consider the task of integrating the blocks in $V_R^{(1)}$. Employing the Weingarten calculus of Eqs.~\eqref{eq:Haarmoment}--\eqref{eq:Haar2moment2} we can express the result of $\langle\Delta\Omega_{\vec{q}\vec{p}}^{\vec{q'}\vec{p'}}\rangle_{V_R^{(1)}}$ as a linear combination of all possible contractions of the $w$ and $\LC_1$ wires. Noting that $\Delta\Omega_{\vec{q}\vec{p}}^{\vec{q'}\vec{p'}}=N_1\left(\widetilde{O}_\ell\right )$, with $N_1(\Ot)=T_{\emptyset}(\Ot) - \frac{1}{4} T_{w}(\Ot)$ (see Eqs.~\ref{eq:F-graph} and \ref{eq:T-contractions}), we write
\begin{align} \label{eq:recursive-graph-equation}
    \int_{V_R^{(1)}}\Delta\Omega_{\vec{q}\vec{p}}^{\vec{q'}\vec{p'}}d\mu 
    =\int_{V_R^{(1)}} N_1\left(\widetilde{O}_\ell\right )d\mu
    =
    \sum_{\beta}\sum_{s' \in \PC(\LC_1)} a_{1,\beta,s'} p_{1,\beta,s'} N_{\beta}(\Ot_{\ell-1})\,,
\end{align}
where we denote $\int_{V_R^{(1)}} d\mu= \int ... \int dW_1 ... dW_m$ as the integral with respect to the Haar measure over the unitary group for all gates $\{ W_{\text{init}} \}_{i=1}^m$ appearing in $V_R^{(1)}$. The coefficients $a_{1,\beta,s'}$ are real numbers to be found, the non-zero of which characterize the module that we are integrating. Moreover, we remark that $\PC(\LC_1)$ and $\PC(\LC_1 \cup w)$ respectively denote the power sets of the qubits in $\LC_1$ and $\LC_1 \cup w$. Hence, the summations in~\eqref{eq:recursive-graph-equation} and~\eqref{eq:F-graph} are taken over all the subsets $s'$ and $s$ of $\LC_1$ and $\LC_1 \cup w$, respectively. In addition, here we have defined the {\it nodes}
\begin{align} \label{eq:F-graph}
    N_{\beta}(\Ot_{\ell-1})    =    \sum_{s \in \PC(\LC_1 \cup w)} c_{\beta,s} T_s(\Ot_{\ell-1})\,,
\end{align}
where $c_{\beta,s}$ are real coefficients that characterize the node, and we define the contraction of the operator $\Ot_{\ell-1}$
\begin{align} \label{eq:T-contractions}
    T_s(\Ot_{\ell-1})
    =
    \Tr\left[
    \Tr_{s,\LC_2,...,\LC_\ell,\LCb}\left[(\ketbra{\vec{p}}{\vec{q}}\otimes\id_{w}) \Ot_{\ell-1} \right]
    \Tr_{s,\LC_2,...,\LC_\ell,\LCb}\left[(\ketbra{\vec{p'}}{\vec{q'}}\otimes\id_{w}) \Ot_{\ell-1} \right]
    \right]\,.
\end{align}
Here $\ketbra{\vec{p}}{\vec{q}}$ and $\ketbra{\vec{p'}}{\vec{q'}}$ are operators over the qubits in  $S_{\wbar} \setminus \LC_1$. The contribution of the projectors $\ketbra{\vec{p}}{\vec{q}}$ and $\ketbra{\vec{p'}}{\vec{q'}}$ over the qubits in  $\LC_1$ is given by the operators $p_{1,\beta,s}$, which are defined as
\begin{align}\label{eq:pcontraction}
    p_{1,\beta,s}
    &=
    \Tr\left[
    \Tr_{s}\left[\ketbra{\vec{p}}{\vec{q}}_{\LC_1} \right]
    \Tr_{s}\left[\ketbra{\vec{p'}}{\vec{q'}}_{\LC_1} \right]
    \right] \\
    &=
    \delta_{({\vec{p}}{\vec{q}})_{s}} 
    \delta_{({\vec{p'}}{\vec{q'}})_{s}} \delta_{({\vec{p}}{\vec{q'}})_{\sbar}} \delta_{({\vec{p'}}{\vec{q}})_{\sbar}}
\end{align}
where $\sbar=\LC_1 \setminus s$. We refer to reader to Fig.~\ref{fig:toygraph}(a) for a graphical representation of the operator $p_{1,\beta,s}$.

Here we remark that $a_{1,\beta,s'}$ and $c_{\beta,s}$ are not uniquely defined for a given integral, since for every $s'$, we can replace $c_{\beta,s}$ by $\gamma_{\beta} c_{\beta,s}$ and $a_{1,\beta,s'}$ by $a_{1,\beta,s'}/\gamma_{\beta}$, where $\gamma_{\beta}$ can be any non-zero number:
\begin{align} \label{eq:gauge-freedom}
   \sum_{\beta} \sum_{s' \in \PC(\LC_1)} a_{1,\beta,s'} p_{1,\beta,s'} \sum_{s \in \PC(\LC_1 \cup w)} c_{\beta,s} T_s(\Ot_{\ell-1})
    =
   \sum_{\beta} \sum_{s' \in \PC(\LC_1)} \frac{a_{1,\beta,s'}}{\gamma_{\beta}} p_{1,\beta,s'} \sum_{s \in \PC(\LC_1 \cup w)} \gamma_{\beta} c_{\beta,s} T_s(\Ot_{\ell-1})\,.
\end{align}
To fix this freedom, we henceforth impose that $c_{\beta, \emptyset}=1$ or $-1$ and that $a_{1,\beta,s'} \geq 0$. The choice of $a_{1,\beta,s'} \geq 0$ guarantees that the edge coefficients of the GRIM graph are always positive.

\subsubsection{Construction of the graph}

\begin{figure}[t]
    \centering
    \includegraphics[width=.8\columnwidth]{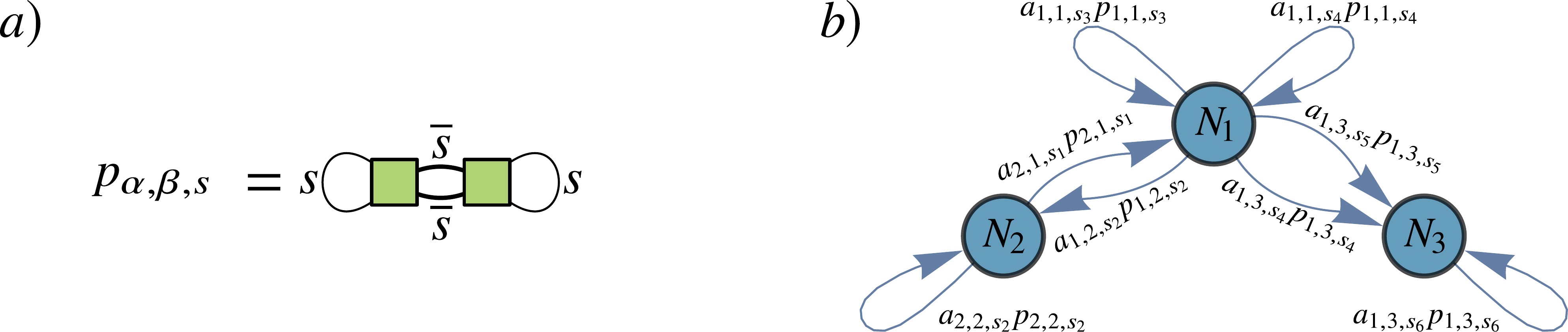}
    \caption{ (a) Schematic representation of the $p_{\alpha,\beta,s}$ coefficients. The contractions of the operators $\ket{\vec{p}}\bra{\vec{q}}$ lead to the terms $p_{\alpha,\beta,s}$ according to~\eqref{eq:pcontraction}. (b) Representative graph obtained after integrating the modules in the forward light-cone $\LC_F$. The nodes $N_\alpha$ is connected to the node $N_\beta$ via an (oriented) arrow. Each edge is associated to an operator $e_{\alpha,\beta,s}=a_{\alpha,\beta,s}p_{\alpha,\beta,s}$.  }
    \label{fig:toygraph}
\end{figure}

Employing the notation introduced in the previous section we can now construct the structure of the graph $\GC_w$, of which we provide a canonical version in Fig.~\ref{fig:toygraph}(b). The recursive procedure is defined as follows:
\begin{itemize}
    \item We define an initial node $N_1$ defined as
    \begin{align} \label{eq:init-node}
        N_1(\Ot)=T_{\emptyset}(\Ot) - \frac{1}{4} T_{w}(\Ot)\,,
    \end{align}
    for any operator $\Ot$. Note that $N_1$ is simply a particular case of Eq. \eqref{eq:F-graph}.  As described in the previous section, we integrate the unitaries in $V_R^{(1)}$ and obtain a result of the form~\eqref{eq:recursive-graph-equation}. In the graph, we represent this by drawing an oriented edge (an arrow) from the initial node $N_1$ to the nodes $N_\beta$ that appear on the right-hand side of~\eqref{eq:recursive-graph-equation}. 
    \item Let $N_\alpha$ be a node in the  graph obtained by integrating the $(i-1)$-th modules. To construct the next nodes in the graph, we integrate $N_\alpha$ over the module $V_R^{(i)}$. Employing Eq.~\eqref{eq:recursive-graph-equation} we can write the result as
    \begin{align}\label{eq:recursionstep}
        \int_{V_R^{(i)}} N_\alpha(\Ot^{(i)})d\mu 
        =
        \sum_\beta \sum_{s' \in \PC(\LC_i)} a_{\alpha,\beta, s'} p_{\alpha,\beta, s'} N_{\beta}(\Ot^{(i+1)})=
        \sum_\beta \sum_{s' \in \PC(\LC_i)} e_{\alpha,\beta, s'} N_{\beta}(\Ot^{(i+1)})\,.
    \end{align}
    As shown in Fig.~\ref{fig:toygraph}(b), for any non-zero $a_{\alpha,\beta, s'} p_{\alpha,\beta, s'}$ we associate a new edge (arrow) from $N_\alpha$ to the nodes in the right-hand-side of~\eqref{eq:recursionstep}. Here we define for convenience the operator
    \begin{equation}\label{eq:eoperator}
            e_{\alpha,\beta, s'}=a_{\alpha,\beta, s'} p_{\alpha,\beta, s'}\,.
    \end{equation}
\end{itemize}
Applying the aforementioned process $\ell$ times, we can obtain the basic structure of the graph $\GC_w$, which is composed of a set of nodes $\{N_\alpha\}$ and associated oriented edges. Let us make two important remarks. First, each edge of the graph is now associated with an operator  $e_{\alpha, \beta,s'}$. Below we show how to obtain the real edge coefficients $\lambda_{\alpha, \beta}$ from the $e_{\alpha, \beta,s'}$. Moreover, we also define an equivalence relation between nodes such that $N_{\alpha_1}$ and $N_{\alpha_2}$ correspond to the same node if they apply the same wire contractions, independent of the qubit that the wire represents.

Next, let us define a walk $\widehat{g}$ through the graph as a sequence of connected edges
\begin{equation}
    \widehat{g} = (e_{\alpha_1, \alpha_2,s'_1}, e_{\alpha_2, \alpha_3,s'_2}, ...)\,,
\end{equation}
such that:
\begin{itemize}
    \item The operator $e_{\alpha_i, \alpha_{i+1},s_i}$ is associated to one of the edges connecting $N_{\alpha_i}$ to $N_{\alpha_{i+1}}$.
    \item The walk starts at the node $N_1$ of~\eqref{eq:init-node}, i.e., $\alpha_1=1$.
    \item As discussed below, in all QCNN graphs  the walk ends at a node $N_1$ evaluated at the operator $O$.
\end{itemize}
We denote the set of all such paths of length $\ell$ as $\widehat{P}_\ell$.

From the ensuing graph and the previous results we can now compute the expectation value $\langle\Delta\Omega_{\vec{q}\vec{p}}^{\vec{q'}\vec{p'}}\rangle_{V_R}$.  Consider a walk $\widehat{g} \in \widehat{P}_\ell$. Then, let us define the operator
\begin{equation}
    Q(\widehat{g}) = \prod_{e_{\alpha, \beta,s} \in \widehat{g}} e_{\alpha, \beta,s}\,, 
\end{equation}
that is, the product of the weights along the walk. We can then write
\begin{equation} \label{eq:sum-over-paths}
    \avg{\Delta\Omega_{\vec{q}\vec{p}}^{\vec{q'}\vec{p'}}}_{V_R} =\int_{V_R^{(\ell)}}\ldots\int_{V_R^{(1)}}\Delta\Omega_{\vec{q}\vec{p}}^{\vec{q'}\vec{p'}}d\mu_1\ldots d\mu_\ell  = N_1(O) \sum_{\widehat{g} \in \widehat{P}_\ell} Q(\widehat{g})\,.
\end{equation}
This equation can simplified by noting that, by definition,
\begin{align}
\begin{split}
    N_1(O) &= T_{\emptyset}(O) - \frac{1}{4} T_{w}(O) \\
    &= 
    \Tr\left[
    \Tr_{\LCb}[(\ketbra{\vec{p}}{\vec{q}}\otimes\id_{w}) O]
    \Tr_{\LCb}[(\ketbra{\vec{p'}}{\vec{q'}}\otimes\id_{w}) O]
    \right]
    -
    \frac{1}{4}
    \Tr[(\ketbra{\vec{p}}{\vec{q}}\otimes\id_{w}) O]
    \Tr[(\ketbra{\vec{p'}}{\vec{q'}}\otimes\id_{w}) O]
    \\
    &=
    \Tr[\ketbra{\vec{p}}{\vec{q}}] 
    \Tr[\ketbra{\vec{p'}}{\vec{q'}}] 
    \left( \Tr[O^2] - \frac{1}{4} \Tr[O]^2 \right)
    \\
    &=
    (\delta_{\vec{p}\vec{q}})_{\LCb}
    (\delta_{\vec{p'}\vec{q'}})_{\LCb}
    \varepsilon_O\,,
\end{split}
\end{align}
where $\varepsilon_O = \Tr[O^2] - \frac{1}{4} \Tr[O]^2$ was defined in the main text. Therefore we have:
\begin{equation} \label{eq:sum-over-paths-final}
    \avg{\Delta\Omega_{\vec{q}\vec{p}}^{\vec{q'}\vec{p'}}}_{V_R} 
    =
    (\delta_{\vec{p}\vec{q}})_{\LCb}
    (\delta_{\vec{p'}\vec{q'}})_{\LCb}\,\,
    \varepsilon_O
    \sum_{\widehat{g} \in \widehat{P}_\ell} Q(\widehat{g}) \,.
\end{equation}

The next step in the GRIM is to simplify the graph by lower bounding Eq.~\eqref{eq:sum-over-paths-final} and  associating to each vertex a real positive coefficient $\lambda_{\alpha \beta}$ in place of each operator $e_{\alpha, \beta,s}$ that appears in $Q(\widehat{g})$. We remark that below we analyze the three possible cases for the placement of $W$.

\subsubsection{Simplification of the graph}

Let us first recall that all the $p_{\alpha, \beta,s}$ obtained from integrating the $k$-th module can be written in terms of the operators
$(\delta_{\vec{p}\vec{q'}})_{\sbar_k} (\delta_{\vec{q}\vec{p'}})_{\sbar_k} (\delta_{\vec{p}\vec{q}})_{s_k} (\delta_{\vec{p'}\vec{q'}})_{s_k}$
where $\sbar_k \cup s_k = \LC_k$ and $\sbar_k \cap s_k = \emptyset$. 
Hence, from~\eqref{eq:sum-over-paths-final} we can write
\begin{align}\label{eq:deltacontractions}
    (\delta_{\vec{p}\vec{q}})_{\LCb}
    (\delta_{\vec{p'}\vec{q'}})_{\LCb}
    \prod_{e_{\alpha, \beta,s} \in \widehat{g}}  p_{\alpha, \beta,s}
    =
    (\delta_{\vec{p}\vec{q'}})_{\Sb_P} (\delta_{\vec{q}\vec{p'}})_{\Sb_P} (\delta_{\vec{p}\vec{q}})_{S_P} (\delta_{\vec{p'}\vec{q'}})_{S_P}\,,
\end{align}
where $S_P$ is the disjoint union of all the $s_k$ and $\LCb$, and where $\Sb_P = \wbar \setminus S_P$.

Combining Eqs.~\eqref{eq:VarcSM}, \eqref{eq:sum-over-paths-final}, and~\eqref{eq:deltacontractions} we find 
\begin{align}
    \sum_{\substack{\vec{p}\vec{q}\\\vec{p'}\vec{q'}}}
    \avg{\Delta\Omega_{\vec{q}\vec{p}}^{\vec{q'}\vec{p'}}}_{V_R}
    \avg{\Delta\Psi_{\vec{p}\vec{q}}^{\vec{p'}\vec{q'}}}_{V_L}
    &=
    \varepsilon_O \sum_{\widehat{g} \in \widehat{P}_\ell}
    \left( \prod_{e_{\alpha, \beta,s} \in \widehat{g}} a_{ \alpha, \beta,s} \right)
    \sum_{\substack{\vec{p}\vec{q}\\\vec{p'}\vec{q'}}} (\delta_{\vec{p}\vec{q'}})_{\Sb_P} (\delta_{\vec{q}\vec{p'}})_{\Sb_P} (\delta_{\vec{p}\vec{q}})_{S_P} (\delta_{\vec{p'}\vec{q'}})_{S_P} 
    \avg{\Delta\Psi_{\vec{p}\vec{q}}^{\vec{p'}\vec{q'}}}_{V_L}
    \\
    &=  \varepsilon_O
    \sum_{\widehat{g} \in \widehat{P}_\ell}
    \left( \prod_{e_{\alpha, \beta,s} \in \widehat{g}} a_{ \alpha, \beta,s} \right)
    \avg{D_{HS}\left(\sigmat_{\Sb_P,w},\sigmat_{\Sb_P} \otimes \frac{\id}{4}\right)}_{V_L}
    \\
    &\geq  \varepsilon_O
    \sum_{\widehat{g} \in \widehat{P}_\ell}
    \left( \prod_{e_{\alpha, \beta,s} \in \widehat{g}} a_{ \alpha, \beta,s} \right)
    \frac{1}{4} \avg{D_{HS}\left(\sigmat_{w}, \frac{\id}{4} \Tr[\sigmat] \right)}_{V_L}
    \frac{1}{2^{|\Sb_P|}}
    \\ \label{eq:one-half-decomposition}
    &= 
    \frac{ \varepsilon_O}{4}
    \avg{\varepsilon_{ \sigmat_{w}}}_{V_L}
    \sum_{\widehat{g} \in \widehat{P}_\ell}
    \left( \prod_{e_{\alpha, \beta,s} \in \widehat{g}} a_{ \alpha, \beta,s} \right)
    \frac{1}{2^{\sum_{e_{\alpha, \beta,s} \in \widehat{g}} |\sbar|}}
    \\ \label{eq:sum-one-half-explicit}
    &=
    \frac{ \varepsilon_O}{4}
    \avg{  \varepsilon_{ \sigmat_{w}} }_{V_L}
    \sum_{\widehat{g} \in \widehat{P}_\ell}
    \left( \prod_{e_{\alpha, \beta,s} \in \widehat{g}} \frac{a_{ \alpha, \beta,s}}{2^{|\sbar|}} \right)\,,
\end{align}
where we denote $|s|$ as the number of qubits in the set $s$. The first equality arises from the fact that  $a_{ \alpha, \beta,s}$ are independent of $\vec{p},\vec{q},\vec{p}',\vec{q}'$. The second equality follows  from Lemma \ref{lemma:s1-s2}. Then, the inequality is obtained by invoking  Lemma \ref{lemma:dhs-inequality} (equivalence of the norms and the monotonicity of the partial trace). In~\eqref{eq:one-half-decomposition} we have used the definition of  $|\Sb_P|$. Finally, in the last line we have simply grouped terms into a single product.

We see from Eq. (\ref{eq:sum-one-half-explicit}) that every coefficient of the form $(\delta_{\vec{p}\vec{q'}})_{\sbar}(\delta_{\vec{p}\vec{q'}})_{\sbar}$ induces a factor $\frac{1}{2^{|\sbar|}}$ in the edge coefficients. That leads us to a new, simplified, version of our graph, where each edge operator $e_{\alpha, \beta,s}$ is replaced by an edge coefficient $\lambda_{\alpha \beta}$:
\begin{align}\label{eq:sumlambda}
    e_{\alpha, \beta,s}
    =
    a_{s_{ \alpha, \beta,s}} p_{ \alpha, \beta,s} 
    \rightarrow
    \lambda_{\alpha \beta}=\sum_s \frac{a_{ \alpha, \beta,s}}{2^{|\sbar|}}\,.
\end{align}
Moreover, two edges connecting two given nodes $N_\alpha$ and $N_\beta$, with associated coefficients $e_{\alpha,\beta,s}$ and $e_{\alpha,\beta,s'}$, can be factored into a single edge in Eq.~\eqref{eq:sum-one-half-explicit}. The associated coefficient is the sum of the $\lambda_{\alpha \beta}$ corresponding to the two edges.

We denote $\GC_w$ as the final graph consisting of the simplified $\lambda_{\alpha \beta}$ coefficients and we denote $P_\ell (\GC_w)$ as the set of paths in this graph. Hence, the final lower bound can be written:
\begin{align} \label{eq:prop1-equation-appendix}
    \sum_{\substack{\vec{p}\vec{q}\\\vec{p'}\vec{q'}}}
    \avg{\Delta\Omega_{\vec{q}\vec{p}}^{\vec{q'}\vec{p'}}}_{V_R}
    \avg{\Delta\Psi_{\vec{p}\vec{q}}^{\vec{p'}\vec{q'}}}_{V_L}
    &\geq
     \frac{ \varepsilon_O}{4}
    \avg{  \varepsilon_{ \sigmat_{w}} }_{V_L}
    \sum_{g \in P_\ell (\GC_w)}
    \Lambda_g\,,
\end{align}
where we define the sequence of connected edges
\begin{equation}
    g = (\lambda_{\alpha_1, \alpha_2}, \lambda_{\alpha_2, \alpha_3}, ...)\,.
\end{equation}
and $$\Lambda_g=\prod_{\lambda_{\alpha \beta} \in g} \lambda_{\alpha \beta}.$$
Note that this constitutes a proof for Proposition~\ref{prop1}.

\subsection{Integration over the unitaries in $V_L$ via the GRIM}

\subsubsection{Method}

In this section we introduce the general method employed to compute the expectation value $\langle  \varepsilon_{ \sigmat_{w}} \rangle_{V_L}$. In contrast to the forward light-cone,  as shown in Fig.~\ref{fig:backward-light-cone}(a), the width of the backward light-cone $\LC_B$ is not bounded and instead grows with the number of layers. The goal here is to employ the GRIM to reduce the number of gates that need to be integrated to compute $\langle  \varepsilon_{ \sigmat_{w}} \rangle_{V_L}$.  

First, let us remark that since we are interested in the operator $\sigmat_{w}=\Tr_{\wbar}[V_L \sigma V_L^\dag]$, all the unitaries in $V_L$ which are not in $\LC_B$ will compile to identity (see Fig.~\ref{fig:backward-light-cone}(a) and (b)). Defining $V_{\LC_B}$ as the unitary containing the blocks in $\LC_B$, then we have that $\langle  \varepsilon_{ \sigmat_{w}} \rangle_{V_L}=\langle  \varepsilon_{ \sigmat_{w}} \rangle_{V_{\LC_B}}$. 

\begin{figure}[t]
    \centering
    \includegraphics[width=.99\columnwidth]{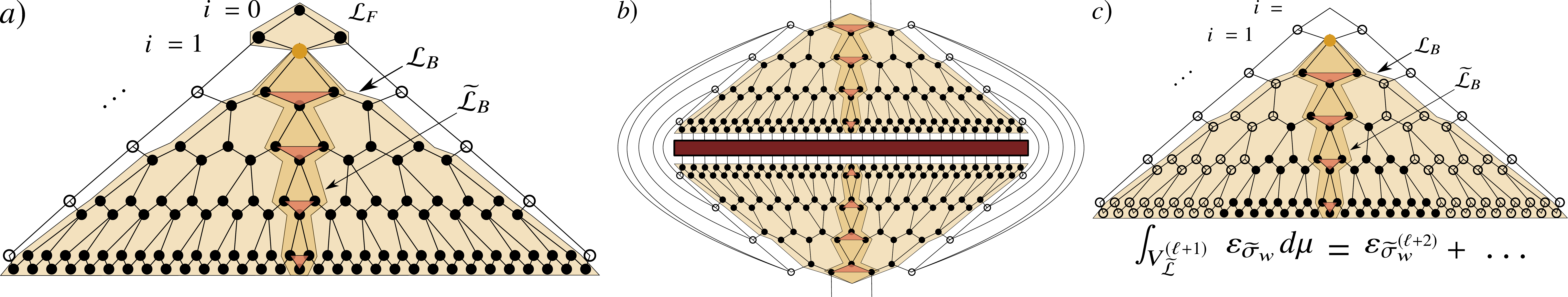}
    \caption{ (a) Schematic representation of the backwards light cone $\LC_B$ (shaded in yellow) and the effective light-cone $\LCt_B$ (shaded in a darker yellow). Here, the width of $\LC_B$ increases with the number of layers. To avoid integrating all the unitaries in $\LC_B$ we define an effective light-cone $\LCt_B$ which can be covered by modules according to the GRIM. (b) Tensor network representation of the unitary $\widetilde{\sigma}_w$. The large red rectangular tensor corresponds to $\sigma$.   Here we can see that the unitaries in $V_L$ which are not in $\LC_B$ simplify to identity when computing $\Tr_{\overline{w}}[V_L \sigma V_L\ad]$. We mark such identities with an unshaded circle.   (c) Integrating the unitaries in $V_{\LCt_B}^{(L-1)}$ leads to a summation of terms according to~\eqref{eq:recursionstepback}. Here, one of these terms corresponds to $\varepsilon_{\widetilde{\sigma}_w^{(\ell+2)}}$. As schematically shown, in the computation of  $\varepsilon_{\widetilde{\sigma}_w^{(\ell+2)}}$ there are many gates that simplify to identity. These gates are indicated here as shallow circles.     }
    \label{fig:backward-light-cone}
\end{figure}

As depicted in Fig.~\ref{fig:backward-light-cone}(a), the next step is to draw an {\it effective backwards light-cone} $\LCt_B$ and to define the unitary $V_{\LCt_B}$ which contains all the gates in $\LCt_B$. This effective light-cone consists of the repetition of $M_\CC$ modules, starting after the integration of an initial unitary if we are in the second sub-layer, as shown in Fig. \ref{fig:initial-backward-first} and \ref{fig:initial-backward-edge-middle}. We will describe below a method to only have to effectively keep track of the reduced light-cone $\LCt_B$ by sequentially removing gates outside of $\LCt_B$ that compile to identity.

We here follow again the first step of GRIM to decompose $V_{\LCt_B}$ into a series of $(L-\ell)$ modules:
\begin{align}\label{eq:unitdecomp}
    V_{\LCt_B} = V_{\LCt_B}^{(\ell+1)} \cdots V_{\LCt_B}^{(L-1)} V_{\LCt_B}^{(L)}\,.
\end{align}
Equation~\eqref{eq:unitdecomp} also allows us to define the operators
\begin{equation}
    \widetilde{\sigma}^{(\ell+k)}=V_{\LCt_B}^{(\ell+k)} \cdots  V_{\LCt_B}^{(L)}{\sigma} \left(V_{\LCt_B}^{(L)}\right)\ad\cdots \left(V_{\LCt_B}^{(\ell+k)}\right)\ad\,,
\end{equation}
and we write $\widetilde{\sigma}_w^{(\ell+k)}=\Tr_{\wbar}[\widetilde{\sigma}^{(\ell+k)}]$.

As usual, the second step of the GRIM is to integrate the unitaries in $V_{\LCb_B}^{(\ell')}$, which again will lead to a recursion of the form
    \begin{align}\label{eq:recursionstepback}
        \int_{V_{\LCt_B}^{(\ell+1)}} \widetilde{N}_\alpha(\widetilde{\sigma}_w^{(\ell')})d\mu 
        =
        \sum_\beta \sum_{s' \in \PC(\LC_i)} a_{\alpha,\beta, s'}  \widetilde{N}_{\beta}(\widetilde{\sigma}_w^{(\ell'+1)})\,,
    \end{align}
where here the nodes $\widetilde{N}_{\alpha}$, contain no $\ket{\vec{p}}\bra{\vec{q}}$ operators due to the fact that $ \varepsilon_{ \sigmat_{w}}$ is independent of $\vec{p}$, $\vec{q}$, $\vec{p}'$, and $\vec{q}'$. Moreover, as explicitly shown below for all three possible scenarios for placement of $W$, there will always be specific tensor contractions  $N_\alpha$ where the unitaries not in  $V_{\LCt_B}$ compile to identify. Such a procedure is schematically shown in Fig.~\ref{fig:backward-light-cone}(c) for the middle module, where computing $\varepsilon_{\sigma_w^{(\ell+2)}}$ leads to many more blocks in $V_{\LC_B}$ being simplified and compiling to the identity. Hence, by sequentially repeating this procedure we  can obtain a lower bound for $\langle  \varepsilon_{ \sigmat_{w}} \rangle_{V_L}$ without needing to integrate all gates in $\LC_{B}$, but just those in $\LCt_B$.

In the next section we explicitly show this procedure for all possible choices of $W$ gate placement.

\subsubsection{$W$ in the first-sublayer}

\begin{figure}[t]
    \centering
    \includegraphics[width=1\columnwidth]{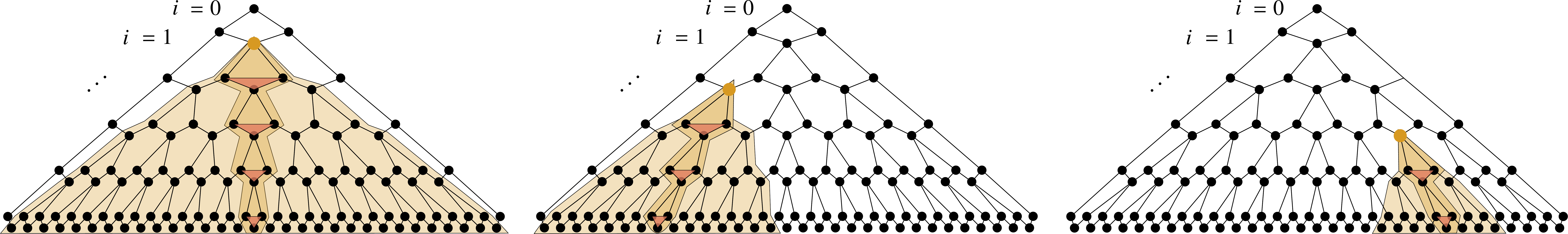}
    \caption{Effective backwards light cone $\LCt_B$ for the case when $W$ is in the first sub-layer or the $\ell$-th layer of the QCNN. Here, $\LCt_B$ can be covered by $(L-\ell)$ modules $M_\MC$.}
    \label{fig:initial-backward-first}
\end{figure}

Here we consider the first case when $W$ is in the first sub-layer of the $\ell$-th layer. This case is schematically shown in Fig.~\ref{fig:initial-backward-first}, where we can see that each unitary in $\LCt_B$ can be covered by $(L-\ell)$ modules $M_\CC$. By explicitly integrating the unitaries in the first module $M_\CC$ of $\LCt_B$ we find:
\begin{align}
    \int_{M_{\MC}}  \varepsilon_{\widetilde{\sigma}_w} d\mu 
    =&  \frac{1}{50} D_{HS}\left(\sigmat^{(\ell+2)}_{w},\frac{\id}{4} \Tr[\sigma] \right)+ \sum_{i,j} \widetilde{c}_{ij} D_{HS}\left(\widetilde{T}_i(\sigmat^{(\ell+2)}), \widetilde{T}_j(\sigmat^{(\ell+2)}) \right)\,, \label{eq:backward-integral-dhs} 
\end{align}
where $\widetilde{c}_{i,j}\geq 0$ $\forall i,j$.  Here,  $\widetilde{T}_i(\sigmat^{(\ell+1)}_{w})$ are different contractions obtained from the operator $\sigmat^{(\ell+1)}_{w}$ and $\widetilde{c}_i$ are real positive coefficients. Since the Hilbert-Schmidt distance is always positive, we have that the lower bound always holds:
\begin{align}
   \int_{M_{\MC}} \varepsilon_{\widetilde{\sigma}_w} d\mu
    &\geq
    \frac{1}{50} 
    D_{HS}\left(\sigmat^{(\ell+2)}_{w},\frac{\id}{4}\Tr[\sigma] \right)\\
    &=
    \frac{1}{50} \varepsilon_{\widetilde{\sigma}_w^{(\ell+2)}}\,.
\end{align}

As discussed in the previous section, repeating this procedure $(L-\ell)$ times and recursively applying the lower bound of~\eqref{eq:backward-integral-dhs}  leads to 
\begin{align}
    \avg{\varepsilon_{\widetilde{\sigma}_w}}_{V_L}
    &\geq
    \left(\frac{1}{50} \right)^{L-\ell}
    D_{HS}\left(\sigma_{w},\frac{\id}{4} \Tr[\sigma] \right)\\
     &=\left(\frac{1}{50} \right)^{L-\ell}\varepsilon_{\sigma_w}\,.\label{eq:final1}
\end{align}

Here we remark that combining Eqs.~\eqref{eq:VarcSM}, ~\eqref{eq:prop1-equation-appendix}, and~\eqref{eq:final1} leads to the proof of Theorem~\ref{theo1SM}. In what follows we consider the case when $W$ is in the second sub-layer.

\subsubsection{$W$ in the edge of the  second  sub-layer}

As shown in Fig.~\ref{fig:initial-backward-edge-middle}(a) here we consider the case when the unitary $W$ is either the first or the last unitary of the second sub-layer of the $\ell$-th layer of the QCNN. In order to group the unitaries in $\LCt_B$ into modules $M_{\MC}$ as we did in the previous section, we first have to integrate the unitary  that is adjacent to $W$ and in the first sub-layer.  This unitary is indicated in Fig.~\ref{fig:initial-backward-edge-middle}(a) by a larger circle with a dashed edge. We call the unitary $W_{\text{init}}$ as it is the first unitary that needs to be integrated in $\LCt_B$.

Explicitly, we find 
\begin{align}
    \int_{W_{\text{init}}}  \varepsilon_{\widetilde{\sigma}_w} d\mu
 =&  \frac{1}{5}   D_{HS}\left(\sigmat^{(\ell+2)}_{w_{\text{init}}},\frac{\id}{4} \Tr[\sigma] \right)+ \sum_{i,j} \widetilde{c}_{i,j} D_{HS}\left(\widetilde{T}_i(\sigmat^{(\ell+1)}),\widetilde{T}_j(\sigmat^{(\ell+1)}) \right)\,, \label{eq:backward-integral-dhs2} 
\end{align}
where we now remark that $\widetilde{\sigma}_{w_{\text{init}}}$ is the reduced state on the qubits that $W_{\text{init}}$ acts on, instead of those that $W$ acts on. Equation~\eqref{eq:backward-integral-dhs2} leads to the following lower bound:
\begin{align}
    \int_{W_{\text{init}}}  \varepsilon_{\widetilde{\sigma}_w}d\mu
    &\geq\frac{1}{5}   D_{HS}\left(\sigmat^{(\ell+2)}_{w_{\text{init}}},\frac{\id}{4} \Tr[\sigma] \right)\\
    &=
    \frac{1}{5}
    \varepsilon_{\widetilde{\sigma}_{w_{\text{init}}}^{(\ell+2)}}\,.
\end{align}
Now we have to integrate $(L-\ell)$ middle modules in $\LCt_B$, meaning that we can can simply employ the result  in~\eqref{eq:final1} to obtain 
\begin{align}
    \avg{\varepsilon_{\widetilde{\sigma}_w}}_{V_L}
    \geq
    \frac{1}{5} \left(\frac{1}{50} \right)^{L-\ell}
   \varepsilon_{\sigma_{w_{\text{init}}}}\,. \label{eq:final2}
\end{align}

\begin{figure}[t]
    \centering
    \includegraphics[width=1\columnwidth]{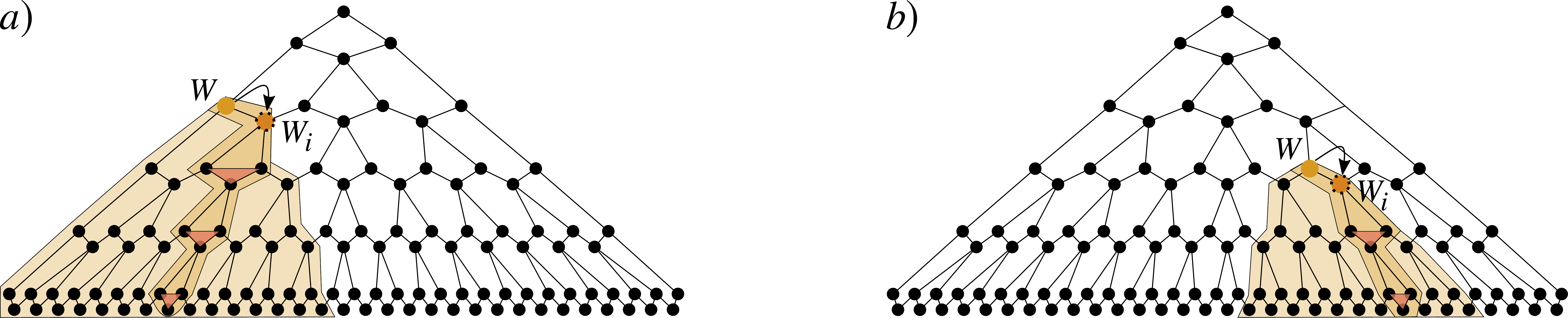}
    \caption{(a) The unitary $W$ is in the edge of the second sub-layer. Here all unitaries in the backwards light-cone $\LCt_B$ can be grouped in modules $M_\MC$ except for the initial unitary $W_{\text{init}}$. (b) The unitary $W$ is in the the second sub-layer but not in its edge. The backwards light-cone structure is similar to the one in panel (a). }
    \label{fig:initial-backward-edge-middle}
\end{figure}

\subsubsection{$W$ in the second sub-layer but not in the edge}

As shown in Fig.~\ref{fig:initial-backward-edge-middle}(b) here we consider the case when the unitary $W$ is in the second sub-layer of the $\ell$-th layer of the QCNN, but is not the first or the last unitary. Similarly to the previous case, we here have to integrate a unitary $W_{\text{init}}$. We find, 
\begin{align}
    \int_{W_{\text{init}}}  \varepsilon_{\widetilde{\sigma}_w} d\mu
 =&  \frac{1}{50}   D_{HS}\left(\sigmat^{(\ell+2)}_{w_{\text{init}}},\frac{\id}{4} \Tr[\sigma] \right)+ \sum_{i,j} \widetilde{c}_{ij} D_{HS}\left(\widetilde{T}_i(\sigmat^{(\ell+1)}),\widetilde{T}_j(\sigmat^{(\ell+1)}) \right)\,, \label{eq:backward-integral-dhs23} 
\end{align}
leading to the following lower bound:
\begin{align}
    \int_{W_{\text{init}}}  \varepsilon_{\widetilde{\sigma}_w} d\mu
    &\geq
    \frac{1}{50}
    D_{HS}\left(\sigmat^{(1)}_{w},\frac{\id}{4} \Tr[\sigma] \right)\\
    &=
    \frac{1}{50}
    \varepsilon_{\widetilde{\sigma}_{w_{\text{init}}}^{(\ell+2)}}\,.
\end{align}
Since the remaining unitaries can be grouped into $M_{\MC}$ modules, we can use Eq.~\eqref{eq:final1}  to obtain 
\begin{align}
    \avg{\varepsilon_{\widetilde{\sigma}_w}}_{V_L}
    \geq
    \frac{1}{50} \left(\frac{1}{50} \right)^{L-\ell}
    \varepsilon_{\sigma_{w_{\text{init}}}}\,. \label{eq:final3}
\end{align}

\subsection{General lower bound on the variance}

Here we combine the results previously obtained to derive a general lower bound for $\Var[\partial_\mu C]$ which is independent of whether $W$ is in the first or in the second sub-layer. 

Let us start by noting that Eqs.~\eqref{eq:final1}, \eqref{eq:final2}, and~\eqref{eq:final3} all provide a lower bound for the expectation value of $\avg{\varepsilon_{\widetilde{\sigma}_w}}_{V_L}$. Specifically, the bound in Eqs.~\eqref{eq:final2}, and~\eqref{eq:final3} can be obtained from that of~\eqref{eq:final1} by respectively dividing by $\frac{1}{5}$ and $\frac{1}{50}$. Taking the minimum of these three lower bounds, i.e.  the one associated with the coefficient $\frac{1}{50}$ gives the general lower bound for the backward light-cone:
\begin{align} \label{eq:backward-light-cone-lower-bound}
    \avg{\varepsilon_{\widetilde{\sigma}_w}}_{V_L}
    \geq
    \left(\frac{1}{50} \right)^{L-\ell+1}
    \varepsilon_{\widehat{\sigma}}\,,
\end{align}
where $\widehat{\sigma}=\sigma_{w}$ if $W$ is in the first sub-layer, or $\widehat{\sigma}=\sigma_{w_{\text{init}}}$ if $W$ is in the second sub-layer. 

Combining Eqs.~\eqref{eq:VarcSM}, ~\eqref{eq:prop1-equation-appendix}, with~\eqref{eq:backward-light-cone-lower-bound} results in a generalized form of Theorem~\ref{theo1SM}
\begin{align}
    \Var[\partial_\mu C]
    &\geq
    \frac{1}{9}\frac{\Tr[H_\mu^2] \varepsilon_O \varepsilon_{ \widehat{\sigma}}}{50^{L-\ell+2}}
    \sum_{\vec{g} \in P_\ell(\GC_w)}
    \Lambda_{\vec{g}}\,.
\end{align}

Moreover, one can further generalize this result by noting that, as shown in Fig.~\ref{fig:initial-backward-trick}, one can sequentially apply the trick of integrating a single $W_{ij}$ unitary in each sub-layer to change how $\LCt_B$ is defined. Now, $\LCt_B$ can be composed of a single unitary per sub-layer, and then of  a sequence of center modules.  This will lead to a final result in which $\widehat{\sigma}$ is the reduced state in {\it any} odd pair of qubits in $\LC_B$.

\begin{figure}[t]
    \centering
    \includegraphics[width=.75\columnwidth]{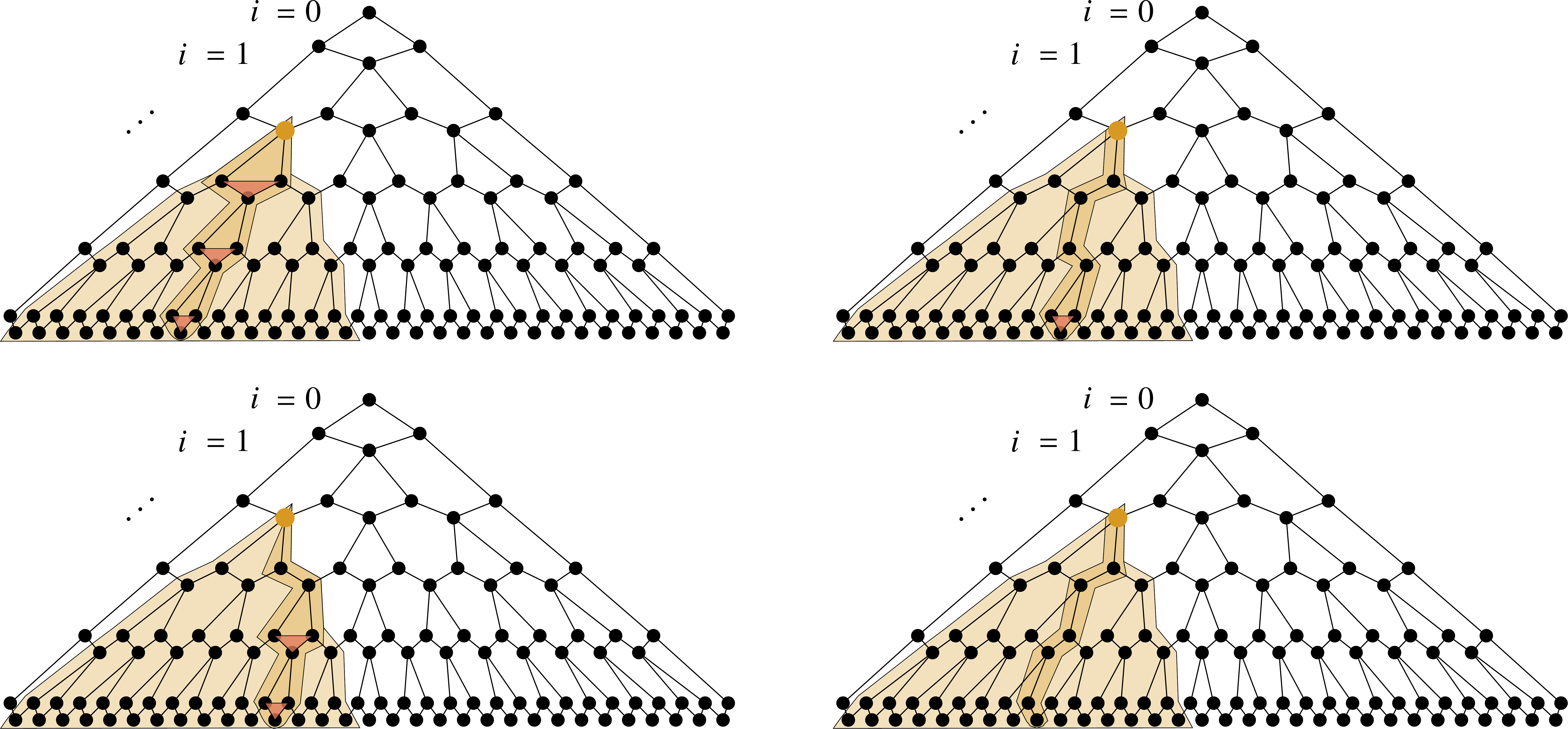}
    \caption{By integrating a single unitary in each sub-layer one can obtain a reduced state $\widehat{\sigma}$ in the sub-space of any odd pair of qubits in $\LC_B$. }
    \label{fig:initial-backward-trick}
\end{figure}

\section{Proof of Corollary~\ref{coro1}}

In this section we present the proof of Corollary~\ref{coro1}, which requires us to analyze the scaling of the term $ \sum_{\vec{g} \in P_\ell(\GC_w)} \Lambda_{\vec{g}}$ in Theorem~\ref{theo1SM}. Specifically, assuming that $L\in\OC(\log(n))$ and that $\Tr[H_\mu^2]\varepsilon_O \varepsilon_{ \widehat{\sigma}} \in\Omega(1/\poly(n))$ we need to show that there is always one path $\vec{g}$ such that the coefficient $\Lambda_{\vec{g}}$ is not exponentially vanishing with the system size, i.e., that $\Lambda_{\vec{g}}\in\Omega(1/\poly(n))$. Moreover, we recall that we have to consider three different cases:
\begin{itemize}
    \item Case 1: The light-cone $\LC_F$ can be covered by $\ell$ modules $M_{\CC}$.
    \item Case 2: The light-cone $\LC_F$ can be covered by $(\ell-1)$ modules $M_{\EC}$ and a module $M_{\CC}$.
    \item Case 3: The light-cone $\LC_F$ can be covered by $\ell'$ modules $M_{\MC}$, $\ell''$ modules $M_{\EC}$ and a module $M_{\CC}$ such that $\ell'+\ell''+1=\ell$.
\end{itemize}

\subsection{Case 1}

Let us recall that the graph for the Case 1 (when $W$ is in the first sub-layer) was presented in the main text. As shown in Fig.~\ref{fig:graphsfinal}(a, top), there is a single path through $\GC_w$ and we have 
\begin{equation}
     \sum_{\vec{g} \in P_\ell(\GC_w)}\Lambda_{\vec{g}}=\left(\frac{28}{125}\right)^\ell\,.
\end{equation}
Then, recalling that $\ell\leq L$, it follows that 
\begin{equation}
    \sum_{\vec{g} \in P_\ell(\GC_w)}\Lambda_{\vec{g}}\geq \left(\frac{28}{125}\right)^L\,,
\end{equation}
where we see that the lower bound is in $\Omega(1/\poly(n))$ if $L\in\OC(\log(n))$.

The case when $W$ is in the second sub-layer leads to an additional module being added to the graph (see Fig.~\ref{fig:graphsfinal}(a, bottom)). This result simply modifies the lower bound as 
\begin{equation}
    \sum_{\vec{g} \in P_\ell(\GC_w)}\Lambda_{\vec{g}}\geq\frac{1}{5} \left(\frac{28}{125}\right)^L\,,
\end{equation}
and we recover again that the lower bound is in $\Omega(1/\poly(n))$ if $L\in\OC(\log(n))$.

\begin{figure}[t]
    \centering
    \includegraphics[width=1\columnwidth]{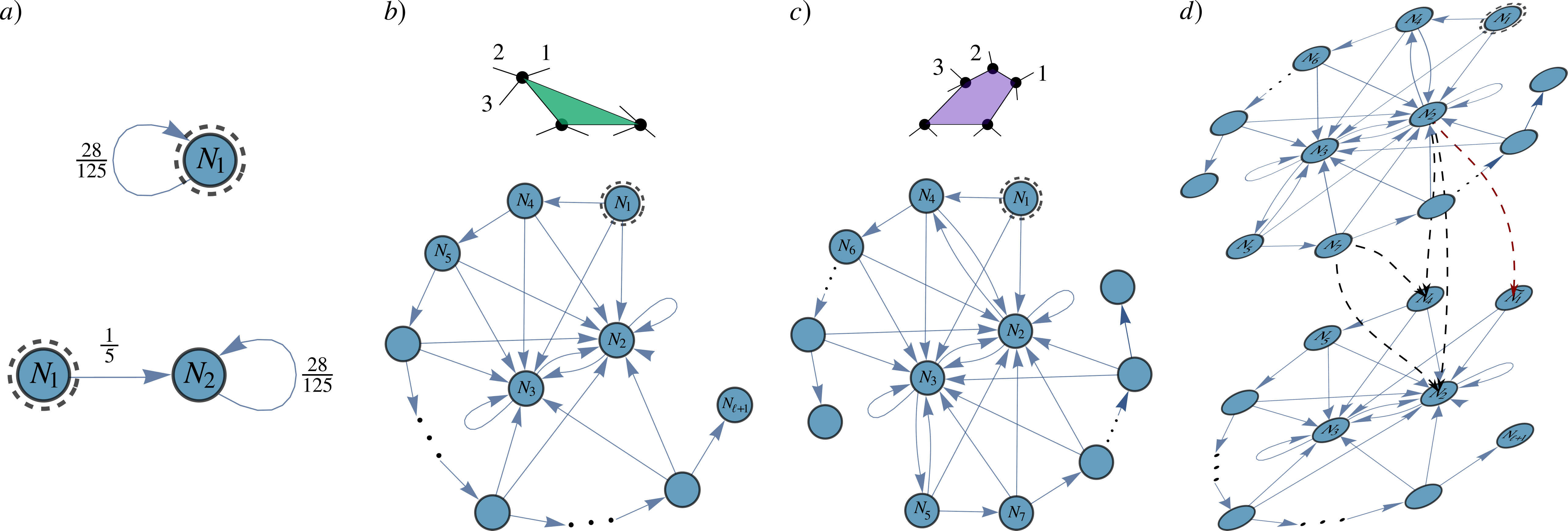}
    \caption{Graphs $\GC_w$ obtained through the GRIM by integrating (a)  $\ell$ center modules $M_\CC$, (b) $\ell$ edge modules $M_\EC$, (c) $\ell$ middle modules $M_\MC$. In panels (b) and (c) we have indicated how the wires have been labeled. (d) Transition between the middle and the edge graphs of (b) and (c). All the nodes in the middle graph are connected to the nodes $N_2$ and $N_4$ on edge graph. The node $N_2$ of the middle graph is also connected to a node $\widehat{N}_1$ that connects to the rest of the edge graph.}
    \label{fig:graphsfinal}
\end{figure}

\subsection{Case 2}

Let us first consider the case when $W$ is in the first sub-layer. The graph obtained after integrating $\ell-1$ modules is shown in Fig.~\ref{fig:graphsfinal}(b). Let us here describe the nodes and their associated coefficients. 

The nodes in Fig.~\ref{fig:graphsfinal}(b) are given by the contractions
\begin{align}
    N_1 &= T_{\{\}} -\frac{1}{4}T_{\{1,2\}} \\
    N_2 &= T_{\{\}} -\frac{1}{2}T_{\{1\}} -\frac{1}{2}T_{\{2,3\}} +\frac{1}{8}T_{\{1,2,3\}} \\
    N_3 &= -T_{\{\}} +\frac{1}{2}T_{\{1\}} +4T_{\{2,3\}} -2T_{\{1,2,3\}} \\
    N_4 &= -T_{\{\}} +8T_{\{1\}} +\frac{1}{2}T_{\{2,3\}} -2T_{\{1,2,3\}} \\
    N_5 &= T_{\{\}} +2T_{\{1\}} -\frac{1}{2}T_{\{2,3\}} -\frac{1}{2}T_{\{1,2,3\}} \\
    N_{5+k} &= T_{\{\}} +\frac{2(4-k)}{4k-1}T_{\{1\}} -\frac{1}{4}T_{\{2,3\}} -\frac{4-k}{2(4k-1)}T_{\{1,2,3\}}\,.
\end{align}
In Fig.~\ref{fig:graphsfinal}(b) we have shown how the wires are labeled.
Similarly, the edge coefficients $\lambda_{\alpha,\beta}$ are presented in Table~\ref{table:a-coefficients-edge-module}. 
\renewcommand{\arraystretch}{1.75}
\begin{table}[h]
\begin{tabular}{|c|c|c|c|c|c|c|}
\hline
\diagbox{$\alpha$}{$\beta$} & 1 & 2 & 3                                             & 4 & $5$ & $5+k+1$ \\ \hline
1        & 0 & $\frac{4}{375}$ & $\frac{2}{375}$ & $\frac{1}{75}$ & 0 & 0 \\ \hline
2        & 0 & $\frac{24}{125}$ & $\frac{1}{25}$ & 0 & 0 & 0 \\ \hline
3        & 0 & $\frac{32}{125}$ & $\frac{4}{25}$ & $\frac{4}{25}$ & 0 & 0 \\ \hline
4        & 0 & $\frac{32}{125}$ & $\frac{4}{25}$ & 0 & $\frac{12}{25}$ & 0 \\ \hline
$5+k$      & 0 & $\frac{32(k+1)}{125(4k-1)}$ & $\frac{4(k+1)}{125(4k-1)}$ & 0 & 0 & $\frac{16(k+1)-4)}{125(4k-1)}$ \\ \hline
\end{tabular}
\caption{Edge coefficients $\lambda_{\alpha,\beta}$ for the edge module graph.}
\label{table:a-coefficients-edge-module}
\end{table}

Finally, for Case 2, the final module that one needs to integrate is a center module $M_\CC$. Surprisingly, we find that that for all the nodes $N_\alpha$ previously presented we have:
\begin{align}
  \int_{M_\MC} N_2(\widetilde{O}_1) d\mu   &= \frac{72}{125}\varepsilon_O  \label{eq:vareps1}\\
    \int_{M_\CC} N_3(\widetilde{O}_1) d\mu  &=  \frac{48}{125} \varepsilon_O   \\
   \int_{M_\CC} N_4(\widetilde{O}_1) d\mu  &=  \frac{528}{125} \varepsilon_O \\
   \int_{M_\CC} N_5(\widetilde{O}_1) d\mu  &=  \frac{272}{125} \varepsilon_O \\
    \int_{M_\CC} N_{5+k}(\widetilde{O}_1) d\mu  &=  \frac{528}{125(4k -1)} \varepsilon_O \:\ \forall k\geq 1\,,\label{eq:vareps2}
\end{align}
that is, all integrals generate positive coefficients.

From here we can see that one can always find a path $\hat{\vec{g}}$ such that $\Lambda_{\hat{\vec{g}}}$ is a product of numbers $\lambda_i$ in $(0,1)$, meaning that 
\begin{align}
    \sum_{\vec{g} \in P_\ell(\GC_w)}\Lambda_{\vec{g}}&\geq \Lambda_
    {\hat{\vec{g}}}\\
    &=\prod_{i=1}^L\lambda_i \\
    &\geq \left(\lambda_{\min}\right)^L\,,
\end{align}
where $\lambda_{\min}$ is the smallest $\lambda_i$ in the path $\hat{\vec{g}}$. Hence, the lower bound is in $\Omega(1/\poly(n))$ if $L\in\OC(\log(n))$.

When $L$ is in the second sub-layer we simply find that the graph in Fig.~\ref{fig:graphsfinal}(b) is slightly modified, but the scaling of the lower bound for $\sum_{\vec{g} \in P_\ell(\GC_w)}\Lambda_{\vec{g}}$ remains unchanged.

\subsection{Case 3}

Let us first consider the integration of the $\ell'$ middle modules $M_{\MC}$. The graph obtained for for the case when $W$ is in the first sub-layer is presented in Fig.~\ref{fig:graphsfinal}(c). 

The nodes in Fig.~\ref{fig:graphsfinal}(c) are given by the contractions
\begin{align}
    N_1 &= T_{\{\}} -\frac{1}{4}T_{\{1,2\}} \\
    N_2 &= -T_{\{\}} +\frac{1}{10}T_{\{1\}}
    +\frac{8}{5}T_{\{3\}}
    +\frac{1}{5}T_{\{1,2\}}
    +\frac{16}{5}T_{\{2,3\}} -2T_{\{1,2,3\}} \\
    N_3 &= T_{\{\}} -\frac{1}{10}T_{\{1\}}
    -\frac{1}{10}T_{\{3\}}
    -\frac{1}{5}T_{\{1,2\}}
    -\frac{1}{5}T_{\{2,3\}} +\frac{1}{8}T_{\{1,2,3\}} \\
    N_4 &= -T_{\{\}} +\frac{8}{5}T_{\{1\}}
    +\frac{1}{10}T_{\{3\}}
    +\frac{16}{5}T_{\{1,2\}}
    +\frac{1}{5}T_{\{2,3\}} -2T_{\{1,2,3\}} \\
    N_{2k+3} &= T_{\{\}} +\frac{4(2b_k-a_k)}{5(8a_k-b_k)}T_{\{1\}}
    -\frac{1}{10}T_{\{3\}}
    +\frac{8(2b_k-a_k)}{5(8a_k-b_k)}T_{\{1,2\}}
    -\frac{1}{5}T_{\{2,3\}} -\frac{2b_k - a_k}{8a_k-b_k}T_{\{1,2,3\}}\:\ ,\:\ \forall k\geq 1 \\
    N_{2k+4} &= T_{\{\}} +\frac{2(4a_k-b_k)}{5(4b_k-a_k)}T_{\{1\}}
    -\frac{1}{10}T_{\{3\}}
    +\frac{4(4a_k-b_k)}{5(4b_k-a_k)}T_{\{1,2\}}
    -\frac{1}{5}T_{\{2,3\}} -\frac{4a_k-b_k}{8b_k-a_k}T_{\{1,2,3\}}\:\ ,\:\ \forall k\geq 1\,.
\end{align}
In Fig.~\ref{fig:graphsfinal}(c) we show how the wires are labeled. The edge coefficients $\lambda_{\alpha,\beta}$ are presented in Table~\ref{table:a-coefficients-middle-module}, where the coefficients $a_k$, $b_k$, $p_k$ and $q_k$ follow a recursive formula of the form
\begin{equation}
    f_n=2f_{n-1}+f_{n-1}\,,
\end{equation}
with
\begin{align}
    a_0&=1, \quad a_1=1, \quad a_2=3\,,\\
    b_0&=0, \quad b_1=2, \quad b_2=4\,,\\
    p_0&=1, \quad p_1=7, \quad p_2=15\,,\\
    q_0&=1, \quad q_1=4, \quad q_2=11\,.
\end{align}
Hence it is clear that $\lambda_{\alpha,\beta}\geq 0$ for all $\alpha$, and $\beta$.

\renewcommand{\arraystretch}{1.75}
\begin{table}[h]
\begin{tabular}{|c|c|c|c|c|c|c|c|c|}
\hline
\diagbox{$\alpha$}{$\beta$} & 1 & 2 & 3                                             & 4 & 5 & 6 & $2k+5$ & $2k+6$ \\ \hline
1        & 0 & $\frac{1}{750}$ & $\frac{56}{1875}$ & $\frac{1}{750}$ & 0 & 0 & 0 & 0 \\ \hline
2        & 0 & $\frac{2}{125}$ & $\frac{288}{3125}$ & $\frac{8}{625}$ & 0 & 0 & 0 & 0 \\ \hline
3        & 0 & $\frac{272}{3125}$ & $\frac{1}{250}$ & 0 & $\frac{6}{625}$ & 0 & 0 & 0 \\ \hline
4        & 0 & $\frac{9}{250}$ & $\frac{488}{3125}$ & 0 & 0 & $\frac{56}{625}$ & 0 & 0 \\ \hline
$2k+3$      & 0 & $\frac{b_k+4b_{k+1}}{250(8a_k-b_k)}$ & $\frac{8}{3125} \left( \frac{8(17a_k+q_k)}{8a_k-b_k} + \frac{25 a_k}{4b_k-a_k}+\frac{20 a_{k+1}}{4b_k-a_k} \right)$  & 0 & 0 & & $\frac{8(8a_{k+1}-b_{k+1})}{625(8a_k-b_k)}$ & 0 \\ \hline
$2k+4$      & 0 & $\frac{4a_{k+1}+a_k}{250(4b_k-a_k)}$ & $\frac{8}{3125}\frac{25a_k+2500a_{k+1}+68b_k+16p_k}{4b_k-a_k}$ & 0 & 0 & & 0 & $\frac{8(4b_{k+1}-a_{k+1})}{625(4b_k-a_k)}$\\ \hline
\end{tabular}
\caption{Edge coefficients $\lambda_{\alpha,\beta}$ for the middle module graph and for $k\geq 1$.}
\label{table:a-coefficients-middle-module}
\end{table}

Unlike Case 2, here we have to also take into account that after integrating the $\ell'$ middle modules $M_\MC$, one still needs to integrate $\ell''$ edge modules $M_\EC$. As shown in Fig.~\ref{fig:graphsfinal}(d), the result can be obtained by connecting the graph of Fig.~\ref{fig:graphsfinal}(c) to a slightly modified edge graph from Fig.~\ref{fig:graphsfinal}(b). Specifically, the edge graph needs to be modified by replacing the node $N_1$ by a modified node $\widetilde{N}_1$ defined as
\begin{equation}
    \widetilde{N}_1=T_{\{\}} -\frac{1}{2}T_{\{1\}}+\frac{92}{7}T_{\{2,3\}}-\frac{46}{7}T_{\{1,2,3\}}\,.
\end{equation}
Note that $\widetilde{N}_1$ is connected to the nodes $N_2$ and $N_3$ but not to the node $N_4$, meaning that we have to set $\lambda_{1,4}=0$ in the edge graph. Then, as shown in Fig.~\ref{fig:graphsfinal}(d), every node in the middle graph is connected to the nodes $N_2$ and $N_4$ on edge graph. The only exception is the node $N_2$ of the middle graph which additionally is connected to the node $\widehat{N}_1$. In Table~\eqref{table:a-coefficients-transition} we present the transition coefficients $\lambda_{\alpha,\beta}$, where the node $N_\alpha$ belongs to the middle graph, while the node $N_\beta$ to the edge graph. In addition, we also need present the coefficients $\lambda_{\widetilde{1},\beta}$ for the updated  $\widetilde{N}_1$ {\it within} the edge graph: $\lambda_{\widetilde{1},2}=1168/875$, and $\lambda_{\widetilde{1},3}=33/175$. Finally, we remark that Eqs.~\eqref{eq:vareps1}--\eqref{eq:vareps1} remain valid. 

\renewcommand{\arraystretch}{1.75}
\begin{table}[h]
\begin{tabular}{|c|c|c|c|}
\hline
\diagbox{$\alpha$}{$\beta$} & $\widetilde{1}$ & 2                                              & 4 \\ \hline
2        & $\frac{11}{125}$ & $\frac{64}{625}$  & $\frac{4}{625}$  \\ \hline
3        & 0 & $\frac{136}{625}$  & $\frac{6}{625}$  \\ \hline
4        & 0 & $\frac{64}{625}$  & $\frac{44}{625}$  \\ \hline
$2k+3$      & 0 & $\frac{64}{625}\left(\frac{25a_k}{8a_k-b_k}-1\right)$  & $\frac{4}{625}\frac{10(a_k+b_k)+b_{k+1}}{8a_k-b_k}$  \\ \hline
$2k+4$      & 0 & $\frac{32}{625}\frac{11b_k+2p_k}{4b_k-a_k}$  & $\frac{4}{625}\frac{10a_k+a_{k+1}+5b_k}{4b_k-a_k}$  \\ \hline
\end{tabular}
\caption{Transition coefficients $\lambda_{\alpha,\beta}$, where the node $N_\alpha$ belongs to the middle module graph, while the node $N_\beta$ to the edge module graph.}
\label{table:a-coefficients-transition}
\end{table}

Similarly to case $2$ we can see that one can always find a path $\hat{\vec{g}}$ such that $\Lambda_{\hat{\vec{g}}}$ is a product of numbers $\lambda_i$ in $(0,1)$, meaning that 
\begin{align}
    \sum_{\vec{g} \in P_\ell(\GC_w)}\Lambda_{\vec{g}}&\geq \Lambda_{\hat{\vec{g}}}\\
    &=\prod_{i=1}^L\lambda_i \\
    &\geq \left(\lambda_{\min}\right)^L\,,
\end{align}
where $\lambda_{\min}$ is the smallest $\lambda_i$ in the path $\hat{\vec{g}}$. Hence, the lower bound is in $\Omega(1/\poly(n))$ if $L\in\OC(\log(n))$.
Moreover, a similar result can be obtained for $W$ in the second sub-layer. 

Hence, we have shown that Corollary~\ref{coro1} is valid for all possible choices of $W$.

\section{Trainability from pooling}

By consolidating parameters into smaller registers, pooling layers in a QCNN can enhance trainability in comparison to a QNN in which a full $n$-qubit register is maintained throughout a variational circuit. This can be illustrated by the pooling module in Fig.~\ref{fig:poolcnn}, in which an input state $\left( \rho^{(0)}\right)^{\otimes n}$ is processed by pooling layers, locally described by the channel
\begin{equation}
\rho^{(j+1)}=\text{tr}_{A}\left[ I_{j+1}\left(\rho^{(j)}\right)^{\otimes 2} I_{j+1}^{\dagger} \right]
\end{equation}
where \begin{equation}I_{j+1}=\ket{+}\bra{+}_{A}\otimes e^{-i\theta_{+}^{(j+1)}Y} + \ket{-}\bra{-}_{A}\otimes e^{-i\theta_{-}^{(j+1)}Y}\end{equation} is the controlled unitary corresponding to the pooling. Starting from $\rho^{(0)}=\ket{0}^{\otimes n}$ with $n=2^{L}$ for simplicity, a cost function at the $j$-th layer ($j=1,\ldots ,L$) is given by
\begin{align}
C^{(j)}&=1-\text{Tr}\left[\ket{0}\bra{0}^{\otimes {n\over 2^{j}}}\left(\rho^{(j)}\right)^{\otimes {n\over 2^{j}}}\right] \nonumber \\
&= 1-\left( \rho^{(j)}_{0,0} \right)^{n\over 2^{j}}.
\label{eqn:poolcost}
\end{align}

The increase in trainability as $j \rightarrow L$ can be seen clearly by considering the cross-section $\theta_{\pm}^{(j)}=\pm\theta$ for all $j$. For these angles, the cost function $C^{(j)}(\theta_{+},\theta_{-})$ is evaluated in closed form by solving the recursion relation
\begin{equation}
\rho^{(j+1)}_{0,0} -\rho^{(j+1)}_{1,1} = \cos \theta \left( \rho^{(j)}_{0,0} -\rho^{(j)}_{1,1}\right), 
\end{equation}
subject to $\text{Tr}\left[\rho^{(j)}\right]=1$ and $\rho^{(0)}_{0,0}=1$, to obtain $\rho_{0,0}^{(j)}={1+\cos^{j}\theta \over 2}$. Substituting this into (\ref{eqn:poolcost}) results in the closed form
\begin{equation}
C^{(j)}(\theta)=1-\left( {1+\cos^{j}\theta \over 2} \right)^{n\over 2^{\ell}}.\label{eqn:toycost}
\end{equation}
For $\ell=1$, i.e., a single pooling layer, the expression~\eqref{eqn:toycost} is almost everywhere equal to 1 on $[-\pi ,\pi]$ as $n\rightarrow \infty$ and, therefore, $C^{(1)}(\theta)$ exhibits a barren plateau. In contrast, for $j = \log n=L$, one evaluates
\begin{align}
E\left( \Big\vert {dC^{(L)}\over d\theta} \Big\vert \right)&= {L \over 2}\int_{-\pi}^{\pi}{d\theta \over 2\pi} \vert \cos^{L-1}\theta \sin \theta \vert \nonumber \\
&= {1\over \pi}
\end{align}
for all $n$, which implies the absence of barren plateaus in the landscape  because the magnitude of the derivative is bounded away from zero in expectation. Uncorrelating the pooling layers by taking the angles defining the conditional unitary $I_{j}$ to be $\theta_{\pm}^{(j)}:= \pm \theta^{(j)}$, $j=1,\ldots ,L$ does not result in barren plateaus for any parameters, as long as the pooling module is sufficiently deep. In particular, the $j$ layer cost function
\begin{equation}
C^{(j)}(\lbrace \theta^{(k)} \rbrace_{k=1}^{j} )= 1-\left({1+ \prod_{k=1}^{j}\cos \theta^{(k)} \over 2}  \right)^{n\over 2^{j}}
\end{equation}
 for maximum depth $j=L$ satisfies
\begin{equation}
E\left( \Big\vert {dC^{(L)}\over d\theta^{(k)}} \Big\vert \right) = {1\over 2n^{\log_{2}(\pi) -1}}
\end{equation}
for all $k$, which does not exhibit the exponentially fast decrease to 0 required by the definition of barren plateaus.

\end{document}